\documentclass[11pt]{article}
\usepackage{}
\usepackage{amsmath,amssymb,amsbsy,amsfonts, amsopn,amstext,amsxtra,euscript,amscd,graphicx}
\usepackage[paper=a4paper, left=2cm, right=2cm, top=2cm, bottom=2cm,
            includefoot, centering]{geometry}
\usepackage{enumerate,float,bm,setspace,boxedminipage,color,rotating}
\usepackage[colorlinks,linkcolor=blue,anchorcolor=blue,citecolor=blue]{hyperref}

\newcommand{\remove}[1]{}

\newtheorem{definition}{Definition}[section]
\newtheorem{lemma}{Lemma}[section]

\newtheorem{theorem}{Theorem}[section]

\newenvironment{proof}{\vspace{8pt}
\noindent{\bf Proof}: }{{\hfill {\large $\Box$}} \vspace{8pt}}

\begin{document}

\bibliographystyle{plain}

\title{Private Outsourcing of Polynomial Evaluation
and Matrix Multiplication using Multilinear
Maps}

\author{Liang Feng Zhang and Rehanehi Safavi-Naini\\
Institute for Security, Privacy and Information Assurance\\
Department of Computer Science\\
University of Calgary\\
\{liangf.zhang,~rei.safavi\}@gmail.com}

\date{}

\maketitle

\begin{abstract}
{\em Verifiable computation} (VC)  allows a computationally weak client to outsource
 the evaluation  of a  function on many inputs to a powerful but untrusted server.
The  client invests a large
amount of off-line computation and gives   an encoding of
its  function to the server.
The server returns both an evaluation of the function on the client's input
and a  proof
such that the client can verify the evaluation using
substantially less effort than  doing the evaluation on its own.
We consider how to privately outsource computations using {\em privacy preserving} VC schemes
whose executions reveal no information on the client's input or function
to the server.
We construct  VC schemes with  {\em input  privacy}
for univariate polynomial evaluation and matrix multiplication and then
extend them such that the {\em function privacy}
is also achieved.
Our  tool is the recently
developed {mutilinear maps}.
The proposed VC schemes can be used in
outsourcing  {private information
retrieval (PIR)}.
\end{abstract}

\section{Introduction}

 The rise of cloud computing in recent years  has made outsourcing of storage and computation a reality
with many cloud service providers offering attractive services.
Large computation can hugely impact resources (e.g. battery) of weak clients.
Outsourcing computation removes this bottleneck but also raises a natural question:
how to assure  the computation is  carried out correctly as the server is untrusted.
This assurance is not only against malicious behaviors  but also
 infrastructure  failures  of the server.
 Verifiable computation (VC) \cite{GGP10}
 provides such   assurances  for
a large class of computation delegation scenarios.
The client in this model invests a  large amount of off-line computation
and generates an encoding of its function $f$.
Given the encoding and any input $\alpha$,
the server  computes  and  responds with  $y$ and a proof that $y=f(\alpha)$.
  The client can verify  if the computation has been carried out correctly using
substantially less effort than computing $f(\alpha)$ on its own.
In particular, the client's off-line computation cost  is {\em amortized}
over the evaluations of $f$ on multiple inputs $\alpha$.

VC schemes were  formally defined  by
Gennaro, Gentry and Parno
 \cite{GGP10} and then constructed   for a variety
 of computations \cite{CKV10,BGV11,PRV12,BF12,PST11,FG12,FG12a}.
We say that  a VC scheme is {\em privacy preserving} if its execution
reveals no information on the client's input or function to
the server.
Protecting  the client's input and function from the server is  an essential requirement in
many real-life scenarios.
For example,  a health professional querying  a database of medical records
may need to protect both the identity and the record  of his patient.
 VC schemes  with input privacy have been considered
 in \cite{GGP10,BF12} where a generic  function is written  as a circuit, and
each gate is evaluated using a fully homomorphic encryption scheme (FHE).
These VC schemes evaluate the outsourced  functions as circuits and are not very efficient. Furthermore, the outsourced functions  are given
to the server in clear and therefore  the function privacy is not achieved.
Benabbas, Gennaro and Vahlis \cite{BGV11} and several other works
\cite{FG12,FG12a,PST11} design  VC schemes for specific functions without using FHE.
Some of them even achieve function privacy.
However, they do not
consider the input privacy.

\subsection{Results and Techniques}
\label{section:results}

In this paper, we consider {privacy preserving} VC schemes for specific function evaluations
without using FHE.
The function evaluations  we study include
 univariate polynomial evaluation and matrix multiplication.
Our privacy definition is
  indistinguishability based and  guarantees  no
untrusted server can distinguish between different inputs or functions of the client.
In privacy preserving VC  schemes  both the client's input and function
must be hidden (e.g., encrypted)
from the server  and
   the server must  evaluate the hidden function
   on the hidden input and then  generate a proof that the evaluation
has been carried out correctly.
   Note that such a proof  can be generated using the
    {non-interactive} proof or argument systems
    from  \cite{Mic94,BCCT12} but they  require the
    use of either {random oracles} or {knowledge of exponent} (KoE) type assumptions,
 both of which are considered
{as strong \cite{PST11}} and  have been   carefully avoided  by the VC literatures \cite{GGP10,BGV11,PRV12}.

We construct  VC schemes for univariate polynomial evaluation and
 matrix multiplication that achieve  input privacy
and then extend them such that the function privacy is also achieved.
Our main tool  is the  multilinear maps \cite{GGH12,SW12}.
Very recently, Garg, Gentry, and Halvei \cite{GGH12}
proposed  a
candidate mechanism that would approximate or be the moral equivalent of
multilinear maps for many applications.
In \cite{SW12}, it was believed  that the mechanism opens an exciting opportunity to
study new constructions using a multilinear map abstraction.
Following  \cite{SW12}, we use a framework of leveled multilinear maps
where one can call a group
generator ${\cal G}(1^\lambda,k)$ to obtain
a sequence of groups $G_1,\ldots, G_k$ of order
$N$ along with their  generators $g_1,\ldots, g_k$, where
$N=pq$ for two $\lambda$-bit primes $p$ and $q$.
Slightly abusing notation, if  $i+j\leq k$,
 we can compute a bilinear map operation
on $g_i^a\in G_i, g_j^b\in G_j$ as $e(g_i^a, g_j^b)=g_{i+j}^{ab}$.
 These maps can be seen as implementing
a {\em $k$-multilinear map}.
%
%
%
%
We denote by
\begin{equation}
\label{equation:gammak}
\Gamma_k=(N,G_1,\ldots, G_k, e, g_1,\ldots, g_k)\leftarrow\mathcal{G}(1^\lambda, k)
\end{equation}
 a random $k$-multilinear map instance, where $N=pq$ for two
 $\lambda$-bit primes $p$ and $q$.
 We start with the BGN encryption scheme (denoted by ${\rm BGN}_{2}$)
 of Boneh, Goh and Nissim \cite{BGN05}
which is based on $\Gamma_2$ and
semantically secure when the subgroup decision assumption (abbreviated as SDA,
see Definition \ref{definition:sda}) for $\Gamma_2$ holds.
It is well-known that  ${\rm BGN}_2$ is both additively homomorphic and
multiplicatively homomorphic, i.e.,  given ${\rm BGN}_2$
ciphertexts  ${\sf Enc}(m_1)$ and ${\sf Enc}(m_2)$ one can easily compute
 ${\sf Enc}(m_1+m_2)$ and ${\sf Enc}(m_1m_2)$.
 Furthermore, ${\rm BGN}_2$ supports an unlimited number of
 additive homomorphic operations:  for any integer
 $k\geq 2$, given ${\rm BGN}_2$ ciphertexts ${\sf Enc}(m_1), \ldots, {\sf Enc}(m_k)$
 one can easily compute ${\sf Enc}(m_1+\cdots +m_k)$.
As a result, one can easily
compute ${\sf Enc}(f(\alpha))$ from ${\sf Enc}(\alpha)$ for any
quadratic polynomial $f(x)$.
 On the other hand, ${\rm BGN}_2$ supports only one
 multiplicative homomorphic operation: one cannot compute ${\sf Enc}(m_1m_2m_3)$
 from ${\sf Enc}(m_1), {\sf Enc}(m_2)$ and ${\sf Enc}(m_3)$. In particular,
 one cannot compute ${\sf Enc}(f(\alpha))$ from ${\sf Enc}(\alpha)$ for any
 polynomial $f(x)$ of degree $\geq 3$.
In Section \ref{section:BGN},
we  introduce  ${\rm BGN}_{k}$, which is  a generalization of
${\rm BGN}_2$ over  $\Gamma_k$  and semantically secure
under  the SDA for $\Gamma_k$.
${\rm BGN}_k$   supports both an unlimited number of
additive homomorphic operations and up to $k-1$ multiplicative homomorphic
operations.  As a result, it allows us to  compute
${\sf Enc}(f(\alpha))$ from ${\sf Enc}(\alpha)$ for any degree-$k$
 polynomial $f(x)$.
In  our VC schemes, the client's input and function are encrypted using
${\rm BGN}_{k}$ for a suitable $k$ and the server computes on
the ciphertexts.

\vspace{1mm}
 \noindent
{\bf Polynomial evaluation.}
In Section \ref{section:PE} we  propose a
VC scheme $\Pi_{\rm pe}$  with input privacy (see {Fig.} \ref{figure:vc-pe})
that allows
 the client to outsource the
evaluation of a degree $n$ polynomial $f(x)$ on any input $\alpha$ from a polynomial size domain $\mathbb{D}$.
In \cite{KZG10},
the algebraic property that  there is
a polynomial $c(x)$ of degree $n-1$ such that
 $f(x)-f(\alpha)=(x-\alpha)c(x)$ was applied to construct polynomial commitment schemes.
 Those schemes actually give us a basic VC scheme for univariate polynomial evaluation but  without input privacy.
  Let $e: G_1\times G_1\rightarrow G_2$ be a bilinear map, where $G_1$ and $G_2$
  are cyclic groups of prime order $p$ and $G_1$ is generated by $g_1$.
In the basic VC scheme, the client  makes public
  $t=g_1^{f(s)}$  and  gives $pk=(g_1,g_1^s,\ldots, g_1^{s^n},f(x))$ to the server, where
$s$ is uniformly chosen from $\mathbb{Z}_p$.
To verifiably compute $f(\alpha)$, the client gives
$\alpha$ to the server  and the server
 returns  $\rho=f(\alpha)$ along with
a proof  $\pi=g_1^{c(s)}$.
Finally the client verifies if  $e(t/g_1^\rho,g_1)=e( g_1^s/g_1^\alpha,\pi)$.
The basic VC scheme is secure under the SBDH
assumption \cite{KZG10}. It has  been  generalized in  \cite{PST11}
to construct  VC schemes for multivariate polynomial evaluation.

In  $\Pi_{\rm pe}$,   the
$\alpha$ should be hidden  from the server (e.g., the client  gives ${\sf Enc}(\alpha)$ to the
server)
which makes the server's computation of  $\rho$
and $\pi$ (as in the basic VC scheme)  impossible.
Instead, the best one can expect is to compute  a {ciphertext}
$\rho={\sf Enc}(f(\alpha))$ from ${\sf Enc}(\alpha)$ and
$f(x)$.
This can be achieved  if the underlying encryption scheme $\sf Enc$ is an FHE which we want to avoid.
On the other hand, a proof $\pi$ that the computation of $\rho$
has been carried out correctly should  be given to the client. To the best of our knowledge,
for generating such a proof $\pi$, one may adopt the non-interactive proofs or arguments
of \cite{Mic94,BCCT12}  but those constructions require the  use of either
random oracles or KoE type assumptions which we want to avoid as well.
 Our idea is to adopt the multilinear maps \cite{GGH12,SW12}
which allow the server to homomorphically  compute on
${\sf Enc}(\alpha)$  and $f(x)$ and then generate
$\rho={\sf Enc}(f(\alpha))$.
In $\Pi_{\rm pe}$, the client  picks a $(2k+1)$-multilinear map instance $\Gamma$
as in (\ref{equation:gammak}).  It stores $t=g_1^{f(s)}$ and  gives $\vec{\xi}=(g_1,g_1^s,g_1^{s^2}\ldots, g_1^{s^{2^{k-1}}})$ and $ f(x)$ to the server,  where $k=\log \lceil n+1\rceil$.  It also
sets up ${\rm BGN}_{2k+1}$. In order to verifiably compute
$f(\alpha)$, the client gives $k$ ciphtertexts
$\sigma=(\sigma_1,\ldots, \sigma_k)$ to the server and the server returns
$\rho={\sf Enc}(f(\alpha))$ along with a proof
$\pi={\sf Enc}(c(s))$, where $\sigma_\ell={\sf Enc}(\alpha^{2^{\ell-1}})$
for every $\ell\in [k]$.
Note that  $f(\alpha)$
and $c(s)=(f(s)-f(\alpha))/(s-\alpha)$
are  both polynomials in $\alpha$ and $s$.
In Section \ref{section:BGN}, we show how the server can compute
$\rho$ and
$\pi$ from  $f(x), \sigma$ and
$\vec{\xi}$.
Upon receiving $(\rho,\pi)$, the client decrypts $\rho$ to $y$ and verifies if $e(t/g_1^{y}  , g_{2k}^{p})=
e(g_1^s/g_1^\alpha, \pi^{p})$.
Finally, we can show the security and privacy
of $\Pi_{\rm pe}$ under the assumptions  $(2k+1,n)$-MSDHS  (see Definition
\ref{definition:msdh}) and SDA (see Definition \ref{definition:sda}).

\vspace{1mm}
 \noindent
{\bf Matrix multiplication.}  In Section \ref{section:MM} we
propose a VC scheme $\Pi_{\rm mm}$ with input privacy (see {Fig.} \ref{figure:vc-mm})
that allows the client to outsource the computation of $Mx$ for any
$n\times n$ matrix $M=(M_{ij})$  and vector $x=(x_1,\ldots,x_n)$.
It is
 based on the {algebraic PRFs with closed form efficiency} (firstly defined  by
\cite{BGV11}).
In Section \ref{section:aprf}, we  present
an  algebraic PRF with closed form efficiency ${\sf PRF}_{\rm dlin}=({\sf KG,F})$
over a trilinear map instance $\Gamma$, where
for any secret key $K$ generated by ${\sf KG}$,  ${\sf F}_K$ is a function with domain
$[n]^2$ and range $G_1$.
In  $\Pi_{\rm mm}$, the client gives both $M$ and its
blinded version $T=(T_{ij})$ to the server, where $T_{ij}=g_1^{p^2aM_{ij}}\cdot {\sf F}_K(i,j)$
for every $(i,j)\in[n]^2$ and $a$ is randomly chosen from $\mathbb{Z}_N$ and fixed
for any $(i,j)\in [n]^2$.
It also sets up ${\rm BGN}_3$.
In order to verifiably compute $Mx$,
the client stores $\tau_i=\prod_{j=1}^n e({\sf F}_K(i,j), g_2^{px_j})$
for every $i\in[n]$, where  $\tau_i$ can be efficiently computed using the
 closed form efficiency property of ${\sf PRF}_{\rm dlin}$.
 It gives the ciphertexts $\sigma=({\sf Enc}(x_1),\ldots, {\sf Enc}(x_n))$
  to  the server
 and the server   returns    $\rho_i={\sf Enc}(\sum_{j=1}^n M_{ij}x_j)$  along with
 a proof $\pi_i=\prod_{j=1}^n e(T_{ij},{\sf Enc}(x_j))$ for every $i\in[n]$.
Upon receiving $\rho=(\rho_1,\ldots, \rho_n)$
and $\pi=(\pi_1,\ldots, \pi_n)$, the client can decrypt $\rho_i$ to $y_i$
and verify if $e(\pi_i, g_1^p)=\eta^{py_i}\cdot \tau_i$ for every $i\in[n]$, where $\eta=g_3^{p^2a}$.
Finally, we can show the security and privacy
of $\Pi_{\rm mm}$ under the assumptions 3-co-CDHS (see Definition \ref{definition:3-co-cdhs}), DLIN (see Definition
\ref{definition:3-co-cdhs}) and SDA.

\vspace{1mm}
 \noindent
{\bf Applications.}
Our VC schemes have   applications
in
outsourcing
private information retrieval  where the client stores a large database
$w=w_1\cdots w_n\in \{0,1\}^n$
with the cloud
and later retrieves a bit  without revealing which bit he is interested in.
Outsourcing PIR has practical significance: for example a health professional that
 stores
a  database of medical records with the cloud may want to privately retrieve the record of
a certain patient.
Our VC schemes provide easy solutions for  outsourcing PIR.
A client with database $w$
can outsource a polynomial $f(x)$ to the cloud using
$\Pi_{\rm pe}$,
where $f(i)=w_i$ for every $i\in [n]$. The client can also represent its database
as a $\sqrt{n}\times \sqrt{n}$ matrix $M=(M_{ij})$ and outsource it to the cloud using $\Pi_{\rm mm}$.
Retrieving
any bit $M_{ij}$ can be  reduced to  computing $Mx$ for a 0-1 vector $x\in \{0,1\}^{\sqrt{n}}$ whose $j$-th
bit is 1 and all other bits are 0.
Our indistinguishability based  definition of input privacy (see {Fig.} \ref{figure:experiment})
guarantees that the server cannot learn which bit the client is interested in.

\vspace{1mm}
\noindent
{\bf Discussions.}
Note that
decrypting  $\rho={\sf Enc}(f(\alpha))$ in $\Pi_{\rm pe}$ requires computing
discrete logarithms (see Section \ref{section:BGN}).
Hence, the $f(\alpha)$  should be from a polynomial-size domain
$\mathbb{M}$ since otherwise the client will not be able to decrypt $\rho$
and then verify its correctness.
In fact, this is an inherent limitation of \cite{BGN05}
and  inherited by the generalized
${\rm BGN}$ encryption schemes.
 However, in Section \ref{section:discussion} we shall see that
 the limitation is only theoretical and does not affect  the
 application of $\Pi_{\rm pe}$
 in outsourcing PIR, where $f(\alpha)\in \{0,1\}$.
 One may also argue that with $f(x)$ and the knowledge of ``$f(\alpha)\in \mathbb{M}$",
the server may learn a polynomial size domain $\mathbb{D}$ where $\alpha$  is
drawn from and therefore guess $\alpha$  with non-negligible probability.
However, due to Definition \ref{definition:privacy}, we  stress that the input
privacy achieved by $\Pi_{\rm pe}$ is indistinguishability based
which   does not contradict to the above argument.
In fact, such a privacy level
 suffices for our applications in outsourcing PIR.
In Section \ref{section:ext},
we  show how to modify $\Pi_{\rm pe}$
 such that $f(x)$ is also hidden  and therefore prevent the
 cloud from learning any information about $\alpha$.
Similar discussions as above are applicable to $\Pi_{\rm mm}$.

\vspace{1mm}
 \noindent
{\bf Extensions.}
In Section \ref{section:ext}, we
 modify  $\Pi_{\rm pe}$ and $\Pi_{\rm mm}$
such that the function privacy  is also achieved.
In the modified schemes $\Pi_{\rm pe}^\prime$ (see {Fig.} \ref{figure:pp-vc-pe})
and $\Pi_{\rm mm}^\prime$ (see {Fig.} \ref{figure:pp-vc-mm}), the  outsourced  functions are
 encrypted and then given to the server.
The basic idea is increasing the multi-linearity by 1 such that
both the server and the client can compute on
encrypted inputs and functions with one more application of the multilinear map
$e$. The modified schemes $\Pi_{\rm pe}^\prime$
and $\Pi_{\rm mm}^\prime$ achieve both input and function privacy.

%

\subsection{Related Work}

Verifiable computation can be traced back to the work on
  {interactive proofs} or arguments  \cite{GMR85,Mic94}.
In the context of VC,
the {\em non-interactive} proofs or arguments are much more desirable and
 have been considered
in \cite{Mic94,BCCT12} for various computations. However, they use
either {random oracles} or KoE type assumptions.

Gennaro, Gentry and Parno \cite{GGP10}
 constructed the first  non-interactive VC schemes
 without using random oracles    or KoE type assumptions.
Their construction is based on the FHE  and  garbled
circuits.
Also based on FHE,
Chung et al. \cite{CKV10} proposed a VC scheme that requires no public key.
Applebaum et al. \cite{AIK10} reduced  VC
to suitable variants of secure multiparty computation
protocols.
  Barbosa et al.
\cite{BF12} also obtained  VC schemes using  delegatable homomorphic encryption.
Although the input privacy has been explicitly considered in \cite{GGP10,BF12},
those schemes evaluate the outsourced functions as circuits and are
not very efficient. Furthermore,
in these schemes the outsourced functions are  known to the
 server. Hence, they
 do not achieve function privacy.

Benabbas et al.  \cite{BGV11} initiated a line of research
on efficient VC schemes for {specific
function (polynomial) evaluations based on
algebraic PRFs with closed form efficiency.
In particular, one of their VC schemes achieves function privacy.
 Parno et al. \cite{PRV12} initiated a line of
research on public VC schemes for evaluating Boolean formulas, where
 the correctness of
the server's computation can be verified by any client.
Also based on   algebraic PRFs with closed form efficiency, Fiore et al. \cite{FG12,FG12a}  constructed
 public VC schemes for both
 polynomial evaluation and matrix
multiplication.
Using the  the idea of  polynomial commitments
 \cite{KZG10}, Papamanthou et al. \cite{PST11} constructed
 public VC schemes that enable {efficient updates}.
 A common drawback of \cite{BGV11,PRV12,FG12,FG12a,PST11} is that
 the input privacy is not achieved.
 VC schemes with other specific properties
  have also been constructed in
   \cite{GKR08,Mic94,CKLR11,BCCT12,CRR11,CKKC12}.
   However, none of them is privacy preserving.

\vspace{1mm}
 \noindent
{\bf Organization.}
In Section \ref{section:preliminary}, we firstly review
several  cryptographic assumptions related to multilinear
 maps;  then introduce a generalization  of
 the  BGN encryption scheme \cite{BGN05};
 we also recall
algebraic PRFs with closed form efficiency and the formal definition of VC.
In Section \ref{section:schemes}, we present
our VC schemes for univariate polynomial evaluation
and matrix multiplication. In Section \ref{section:applications},
we show applications of our VC schemes in outsourcing PIR.
 Section \ref{section:conclusions} contains some concluding remarks.

\section{Preliminaries}
\label{section:preliminary}

For any finite set $A$, the notation $\omega\leftarrow A$
means that $\omega$ is uniformly  chosen from $A$.    Let $\lambda$ be a security
 parameter. We denote by ${\sf neg}(\lambda)$
the class of negligible  functions  in $\lambda$, i.e., for every constant $c>0$, it is less than
$\lambda^{-c}$ as long as $\lambda$ is large enough.
We  denote by ${\sf poly}(\lambda)$ the class of polynomial functions in $\lambda$.

\subsection{Multilinear Maps and Assumptions}
\label{section:multilinear-map}

In this section, we review several cryptographic assumptions concerning
multilinear maps.
Given the $\Gamma_k$ in (\ref{equation:gammak}) and  $x\in G_i$,
 the
 subgroup decision problem in $G_i$  is  deciding whether $x$ is of
 order $p$ or not, where $i\in[k]$.  When $k=2$,  Boneh et al.
  \cite{BGN05} suggested  the
Subgroup Decision Assumption (SDA) which says that the subgroup decision problems
 in $G_1$ and $G_2$ are intractable. In this paper, we make the same assumption but for a general
 integer $k\geq 2$.
 \begin{definition}
 \label{definition:sda}
{\em (${\rm SDA}$)}
 We say that ${\rm SDA}_i$ holds  if for any probabilistic polynomial time {\em (PPT)}
algorithm $\mathcal{A}$,
$|\Pr[\mathcal{A}(\Gamma_k, u)=1]-
\Pr[\mathcal{A}(\Gamma_k,u^q)=1]|
<{\sf neg}(\lambda),$
where the probabilities are taken over $\Gamma_k\leftarrow {\cal G}(1^\lambda,k)$,
 $u\leftarrow G_i$ and ${\cal A}$'s random coins.
 We say that {\em SDA} holds  if
${\rm SDA}_i$ holds for every $i\in [k]$.
\end{definition}
The following lemma shows that SDA is  equivalent to ${\rm SDA}_1$ (see
Appendix \ref{appendix:sda} for the proof).
\begin{lemma}
\label{lemma:sda}
If  ${\rm SDA}_i$ holds, then ${\rm SDA}_j$ holds for every $j=i+1,\ldots,k$.
\end{lemma}

The $k$-Multilinear $n$-Strong  Diffie-Hellman assumption ($(k,n)$-MSDH)
was suggested in \cite{PTT10}: Given
$g_1^s, g_1^{s^2}, \ldots, g_1^{s^n}$ for some $s\leftarrow \mathbb{Z}_N$,
it is difficult for any PPT algorithm to
find $\alpha\in \mathbb{Z}_N\setminus \{-s\}$ and output
$g_k^{1/(s+\alpha)}$.
\begin{definition}
\label{definition:msdh}
{\em ($(k,n)$-MSDH)}
For any PPT
algorithm $\mathcal{A}$,
$\Pr\big[\mathcal{A}(p,q, \Gamma_k, g_1, g_1^s,  \ldots,
g_1^{s^n})=(\alpha, g_k^{\frac{1}{s+\alpha}})\big]\\
<{\sf neg}(\lambda),$
where  $\alpha\in \mathbb{Z}_N\setminus \{-s\}$
and the probability is taken over $\Gamma_k\leftarrow {\cal G}(1^\lambda,k)$,
$s\leftarrow \mathbb{Z}_N$
and ${\cal A}$'s random coins.
\end{definition}
We are able to construct a privacy preserving VC scheme (see Section \ref{section:PE}) for univariate polynomial evaluation which is secure  based on
$(k,n)$-MSDH.
Under the $(k,n)$-MSDH assumption,  the following lemma (see  Appendix \ref{appendix:msdhs} for the proof) shows that either one of the following two problems is difficult for  any
PPT algorithm:
(i) given $g_1, g_1^{s}, \ldots, g_1^{s^n}$ for some $s\leftarrow \mathbb{Z}_N$,
 compute $g_k^{p/s}$;
 (ii) given $g_1, g_1^{s}, \ldots, g_1^{s^n}$ for some $s\leftarrow \mathbb{Z}_N$,
 compute $g_k^{q/s}$.
\begin{lemma}
\label{lemma:msdhs}
If $(k,n)$-MSDH holds, then except for a negligible fraction of the
$k$-multilinear map instances $\Gamma_k\leftarrow {\cal G}(1^\lambda,k)$,
either $\Pr[\mathcal{A}(p,q, \Gamma_k, g_1, g_1^s,  \ldots,
g_1^{s^n})=g_k^{p/s}]
<{\sf neg}(\lambda)$  for any PPT algorithm $\cal A$ or $\Pr[\mathcal{A}(p,q, \Gamma_k,  g_1, g_1^s,  \ldots,
g_1^{s^n})= g_k^{q/s}]
<{\sf neg}(\lambda)$  for any PPT algorithm $\cal A$, where the probabilities are taken over
 $s\leftarrow \mathbb{Z}_N$
and ${\cal A}$'s random coins.
\end{lemma}
Due to Lemma \ref{lemma:msdhs}, it looks reasonable to assume that  (i) (resp. (ii)) is  difficult.
Furthermore, under this slightly stronger assumption
(i.e.,  (i) is difficult, called $(k,n)$-MSDHS from now on),
we can construct a VC scheme $\Pi_{\rm pe}$  (see Fig. \ref{figure:vc-pe}) that
is more efficient  than the one based on $(k,n)$-MSDH.
 In Section \ref{section:PE}, we  present the scheme $\Pi_{\rm pe}$
based on  $(k,n)$-MSDHS.
\begin{definition}
\label{definition:msdhs}
{\em ($(k,n)$-MSDHS)}
For any PPT
algorithm $\mathcal{A}$,
$\Pr[\mathcal{A}(p,q, \Gamma_k, g_1, g_1^s,  \ldots,
g_1^{s^n})=g_k^{p/s}]
<{\sf neg}(\lambda),$
where  the probability is taken over $\Gamma_k\leftarrow {\cal G}(1^\lambda,k)$,
$s\leftarrow \mathbb{Z}_N$
and ${\cal A}$'s random coins.
\end{definition}

The
 $k$-Multilinear  Decision Diffie-Hellman assumption ($k$-MDDH) was suggested  in
\cite{GGH12,SW12}: Given
$g_1^s,g_1^{a_1},\ldots, g_1^{a_k}\leftarrow G_1$, it is difficult for any PPT
 algorithm to distinguish between $g_k^{sa_1\cdots a_k}$ and  $h\leftarrow G_k$.
\begin{definition}
\label{definition:decision-k-linear-assumption}
{\em ($k$-MDDH)}
For any
PPT
algorithm $\mathcal{A}$,
$|\Pr[\mathcal{A}(p,q,\Gamma_k,  g_1^s,  g_1^{a_1}, \ldots,  g_1^{a_k},
 g_k^{sa_1\cdots a_k})=1]-
\Pr[\mathcal{A}(p,q,\Gamma_k,g_1^s,  g_1^{a_1},\ldots, g_1^{a_k}, h)=1]|
<{\sf neg}(\lambda),$
where the probabilities are taken over $\Gamma_k\leftarrow {\cal G}(1^\lambda,k)$,  $s, a_1,\ldots, a_k\leftarrow \mathbb{Z}_N,
h\leftarrow G_k$  and ${\cal A}$'s random coins.
\end{definition}
Let $\Gamma_3=(N,G_1,G_2,G_3,e,g_1,g_2,g_3)\leftarrow {\cal G}(1^\lambda,3)$ be a random  trilinear map instance.
 Let  $h_1=g_1^p$ and $h_2=g_2^p$.
The  trilinear co-Computational Diffie-Hellman assumption
for the order $q$ Subgroups (3-co-CDHS) says that
given $h_1^a\leftarrow G_1$ and $h_2^b\leftarrow G_2$, it is difficult for
any PPT algorithm to compute $h_2^{ab}$.
\begin{definition}
\label{definition:3-co-cdhs}
{\em (3-co-CDHS)}
For any PPT algorithm ${\cal A}$,
$
\label{equation:3-linear-cdh}
\Pr[{\cal A}(p,q,\Gamma_3, h_1^a,h_2^b)=h_2^{ab}]
<{\sf neg}(\lambda),
$
where the probability is taken over
$\Gamma_3\leftarrow {\cal G}(1^\lambda,3)$, $a,b\leftarrow \mathbb{Z}_N$ and ${\cal A}$'s random coins.
\end{definition}
Based on the following technical Lemma \ref{lemma:tech},
Lemma \ref{lemma:3-co-cdhs} shows that 3-co-CDHS is not a new assumption but  weaker than 3-MDDH
(see Appendix \ref{appendix:tech} and  \ref{appendix:3-co-cdhs} for their proofs).
\begin{lemma}
\label{lemma:tech}
Let $X$ and $Y$  be two  uniform random variables over  $\mathbb{Z}_N$.
 Then the random variable  $Z= pXY\bmod q$
is {\em statistically close} to the uniform   random variable
$U$ over $\mathbb{Z}_q$, i.e.,  we have that
$\sum_{\omega\in \mathbb{Z}_q}|\Pr[Z=\omega]-\Pr[U=\omega]|<{\sf neg}(\lambda).$
\end{lemma}

\begin{lemma}
\label{lemma:3-co-cdhs}
If {3-MDDH} holds, then  {3-co-CDHS}  holds.
\end{lemma}

The Decision LINear assumption (DLIN) has been suggested in
\cite{BBS04} for cyclic groups that admit bilinear maps.
 In this paper, we use the DLIN assumption  on the groups of $\Gamma_3$.
\begin{definition}
\label{definition:dlin}
{\em (DLIN)} Let $G$ be a  cyclic group of order $N=pq$, where $p,q$
are $\lambda$-bit primes.
For any PPT algorithm ${\cal A}$,
$
|\Pr[{\cal A}(p,q,u, v, w,u^a,v^b,  w^{a+b})=1]-
 \Pr[ {\cal A}(p,q,u,v,w,u^a,v^b,w^{c})=1]|
<{\sf neg}(\lambda),
$
where the probabilities are taken over  $u,v,w\leftarrow G$,
$a,b,c\leftarrow \mathbb{Z}_N$ and ${\cal A}$'s random coins.
\end{definition}

\subsection{Generalized BGN Encryption}
\label{section:BGN}

${\rm BGN}_2$  \cite{BGN05}
allows one to  evaluate quadratic polynomials on encrypted inputs (see Section
\ref{section:results}).
Boneh et al. \cite{BGN05} noted that  this property arises from the bilinear map
and a $k$-multilinear map would enable the evaluation of degree-$k$ polynomials
on encrypted inputs.
Let $\mathbb{M}$ be a polynomial size domain, i.e. $|\mathbb{M}|={\sf poly}(\lambda)$.
Below we generalize
${\rm BGN}_2$ and define ${\rm BGN}_{k}=({\sf Gen, Enc, Dec})$ for any $k\geq 2$, where
 \begin{itemize}
 \item $(pk,sk)\leftarrow {\sf Gen}(1^\lambda,k)$ is a key generation algorithm. It picks
 $\Gamma_k$ as in (\ref{equation:gammak}) and then outputs
 both a public key $pk=(\Gamma_k,  g_1,h)$ and  a secret key $sk=p$, where
    $h=u^q$ for $u\leftarrow G_1$.
 \item $c\leftarrow {\sf Enc}(pk,m)$ is an encryption algorithm which encrypts any
   message $m\in \mathbb{M}$ as a ciphertext $c=g_1^mh^r\in G_1$, where  $r\leftarrow \mathbb{Z}_N$.
 \item $m\leftarrow {\sf Dec}(sk,c)$ is a decryption algorithm which takes as input
 $sk$ and a ciphertext $c$,  and outputs a message $m\in \mathbb{M}$ such that $c^p=(g_1^p)^m$.
 \end{itemize}
Note that all algorithms above  are defined
over $G_1$ but in general they can be
defined over $G_i$ for any $i\in[k]$.
This can be done by setting $pk=(\Gamma_k, g_i, h)$ and replacing
any occurrence of $g_1$ with $g_i$, where $h=u^q$ for  $u\leftarrow G_i$.
Similar to  \cite{BGN05},  one can show that ${\rm BGN}_{k}$
is semantically secure  under the {\rm SDA}.

Below we  discuss useful properties of ${\rm BGN}_k$.
For
every integer $2\leq i\leq k$, we define a map $e_i: G_1\times \cdots \times G_1\rightarrow G_i$
such that
$e_i(g_1^{a_1},\ldots, g_1^{a_i})=
g_i^{a_1\cdots a_i}$ for any $a_1,\ldots, a_i\in \mathbb{Z}_N$.
Firstly, we shall see that  ${\rm BGN}_{k}$
allows us to compute ${\sf Enc}(m_1\cdots m_k)$ from
${\sf Enc}(m_1),\ldots, {\sf Enc}(m_k)$. Suppose ${\sf Enc}(m_\ell)=g_1^{m_\ell}h^{r_\ell}$ for every $\ell\in[k]$, where
$h=g_1^{q\delta}$ for some $\delta\in \mathbb{Z}_N$ and $r_\ell\leftarrow \mathbb{Z}_N$.
Let $h_k=e_k(h, g_1,\ldots, g_1)=g_k^{q\delta}$. Then
$e_k({\sf Enc}(m_1),\ldots, {\sf Enc}(m_k))=g_k^m h_k^r$ is
a ciphertext  of
$m=m_1\cdots m_k$ in $G_k$, where
 $r=\frac{1}{q\delta}(\prod_{\ell=1}^k (m_\ell+q\delta r_\ell)-m).$

\vspace{1mm}
\noindent
{\bf Computing $\rho$ with reduced multi-linearity level.}
In $\Pi_{\rm pe}$,
the client gives   a polynomial $f(x)=f_0+f_1x+\cdots+f_n x^n$
and $k$ ciphertexts $\sigma=(\sigma_1,\sigma_2, \ldots, \sigma_k)$  of $\alpha,\alpha^2,\ldots, \alpha^{2^{k-1}}$ under
${\rm BGN}_{2k+1}$ to the server and  the server returns
 $\rho={\sf Enc}(f(\alpha))$,
 where  $k=\lceil \log (n+1) \rceil$. Below we  show how to
  compute the $\rho$ using $\sigma$ and $f(x)$.   Suppose
$
\sigma_\ell=g_1^{\alpha^{2^{\ell-1}}}h^{r_\ell}
$
for every $\ell\in [k]$,
where $h=g_1^{q\delta}$ for some $\delta\in \mathbb{Z}_N$
 and $r_\ell\leftarrow \mathbb{Z}_N$.
Clearly, any $i\in \{0,1,\ldots,n\}$  has a binary representation
$(i_1,\ldots, i_k)$  such that
$
i=\sum_{\ell=1}^k i_\ell 2^{\ell-1}.
$
Then
$\alpha^i=
\alpha^{i_1}\cdot (\alpha^2)^{i_2}\cdots (\alpha^{2^{k-1}})^{i_k}$
is the product of $i_1+\cdots+i_k$ elements of $\{\alpha, \alpha^2,\ldots, \alpha^{2^{k-1}}\}$.
For every $\ell\in [k]$, let
$
\phi_{\ell}=
\sigma_\ell$
if $i_\ell=1$ and
$\phi_{\ell}=g_1
$ otherwise.
  Then  $\rho_i\triangleq e_k(\phi_1,\ldots, \phi_k)=g_k^{\mu_i}=g_k^{m}h_k^r$ is a
    ciphertext of $
m=\alpha^i$  under ${\rm BGN}_{2k+1}$,
where  $\mu_i=\prod_{\ell=1}^k (\alpha^{2^{\ell-1}}+
q\delta r_\ell)^{i_\ell}$ and $r=\frac{1}{q\delta}(\mu_i-m).
$
Thus,  $\rho=\prod_{i=0}^n \rho_i^{f_i}$ is a ciphertext
 of $f(\alpha)$ under ${\rm BGN}_{2k+1}$.

\vspace{1mm}
\noindent
{\bf Computing $\pi$ with reduced multi-linearity level.}
In  $\Pi_{\rm pe}$,
 $k+1$ group elements $\vec{\xi}=(g_1, g_1^s,\ldots,
g_1^{s^{2^{k-1}}})$ are also known to the server as part of the public key,
 where   $s\leftarrow \mathbb{Z}_N$.
The server must return
 $\pi={\sf Enc}(c(s))$  as the proof that $\rho={\sf Enc}(f(\alpha))$
 has been correctly computed.
Below  we  show how to compute $\pi$ using
 $\vec{\xi}$ and
${\sigma}$.
Note that
$
c(s)=(f(s)-f(\alpha))/(s-\alpha)=\sum_{i=0}^{n-1}\sum_{j=0}^i f_{i+1}\alpha^j s^{i-j}.
$
It suffices to show how to
compute
$\pi_{ij}\triangleq{\sf Enc}(f_{i+1}\alpha^js^{i-j})$
for every $i\in \{0,1,\ldots, n-1\}$ and $j\in \{0,1,\ldots,i\}$.
Let
$(j_1,\ldots, j_k),~(i_1,\ldots,
i_k)\in \{0,1\}^k$ be the binary representations of
$j$ and $i-j$, respectively.
Let
$
\phi_\ell=\sigma_\ell
$ if $j_\ell=1$ and $\phi_\ell=g_1$ otherwise.
Let $\psi_\ell=g_1^{s^{2^{\ell-1}}}$
if $i_\ell=1$ and $\psi_\ell=g_1$ otherwise.
Then it is easy to see that $\pi_{ij}=
e(e_k(\phi_1,\ldots, \phi_k), e_k(\psi_1,\ldots, \psi_k))
=g_{2k}^{\nu_{ij}}=g_{2k}^m h_{2k}^r$ is a ciphertext of $m=\alpha^j s^{i-j}$, where
$\nu_{ij}=s^{i-j}\prod_{\ell=1}^k (\alpha^{2^{\ell-1}}+q\delta r_\ell)^{j_\ell}, h_{2k}=g_{2k}^{q\delta}$ and $r=\frac{1}{q\delta}(\nu_{ij}-m)$.
Let $\nu=\sum_{i=0}^{n-1}\sum_{j=0}^i f_{i+1}\nu_{ij}$.
Thus,  $\pi\triangleq g_{2k}^\nu=\prod_{i=0}^{n-1}\prod_{j=0}^i \pi_{ij}^{f_{i+1}}={\sf Enc}(c(s))$.

\subsection{Algebraic PRFs with Closed Form Efficiency}
\label{section:aprf}

In $\Pi_{\rm mm}$, the client gives
both a square matrix $M=(M_{ij})$ of order $n$ and its blinded version
$T=(T_{ij})$ to the server. The computation of $T$ requires an
algebraic PRF with closed form efficiency, which has very efficient algorithms for certain
computations on large data.
   Formally, an algebraic PRF with closed form efficiency  is a pair
  ${\sf PRF}=({\sf KG, F})$, where
 ${\sf KG}(1^\lambda, {\sf pp})$
 generates a secret key $K$ from any public parameter ${\sf pp}$ and
 ${\sf F}_K: I\rightarrow G$ is a function with domain $I$ and range $G$
  (both specified by ${\sf pp}$).
We say that ${\sf PRF}$ has  {pseudorandom} property if for
  any ${\sf pp}$ and any PPT algorithm ${\cal A}$,
  it holds that
$
  |\Pr[{\cal A}^{{\sf F}_K(\cdot)}(1^\lambda,{\sf pp})=1]-
  \Pr[{\cal A}^{{\sf R}(\cdot)}(1^\lambda,{\sf pp})=1]|<{\sf neg}(\lambda),
$
  where the probabilities are taken over the randomness of ${\sf KG}, {\cal A}$ and
   the random function ${\sf R}: I\rightarrow G$.
Consider an arbitrary computation $\sf Comp$ that takes as input
$R=(R_1,\ldots, R_n)\in G^n$ and $x=(x_1,\ldots, x_n)$,
 and assume that the best algorithm to compute
 ${\sf Comp}(R_1,\ldots, R_n, x_1,\ldots, x_n)$ takes time $t$.
Let $z=(z_1,\ldots,z_n)\in
I^n$.
We say that $\sf PRF$ has {closed
form efficiency for} $({\sf Comp}, z)$
if there is an efficient algorithm ${\sf CFE}$ such that
$
{\sf CFE}_{{\sf Comp},z}(K,x)={\sf Comp}({\sf F}_K(z_1),\ldots,
{\sf F}_K(z_n), x_1, \ldots, x_n)
$
and its running time is
 $o(t)$.

\vspace{1mm}
\noindent
{\bf A PRF with closed form efficiency.}
    Fiore et al. \cite{FG12}
constructed an algebraic PRF with closed form efficiency ${\sf PRF}_{\rm dlin}$ based  on the DLIN assumption for
the  bilinear groups. We generalize it over trilinear groups.
In the generalized  setting,  $\mathsf{KG}$ generates
 $\Gamma_3\leftarrow{\cal G}(1^\lambda,3)$,
 picks $\alpha_i,\beta_i\leftarrow \mathbb{Z}_N, ~ A_i, B_i \leftarrow G_1$ for every $i\in [n]$,  and
 outputs  $K=\{\alpha_i, \beta_i, A_i, B_i: i\in [n]\}$.
 The  function $\mathsf{F}_K$ maps any pair $(i,j)\in [n]^2$
to ${\sf F}_K(i,j)=A_j^{\alpha_i}B_j^{\beta_i}$.
The closed form efficiency of ${\sf PRF}_{\rm dlin}$ is described as below.
Let  $x=(x_1,\ldots, x_n)\in \mathbb{Z}_N^n$. The computation $\sf Comp$
we consider is computing  $\prod_{j=1}^n {\sf F}_K(i,j)^{x_j}$ for all $i\in[n]$.
Clearly, it requires $\Omega(n^2)$  exponentiations if no ${\sf CFE}$ is available.
However, one can precompute
$A=A_1^{x_1}\cdots A_n^{x_n}$ and $B=B_1^{x_1}\cdots B_n^{x_n}$ and have that
$\prod_{j=1}^n {\sf F}_K(i,j)^{x_j}=A^{\alpha_i}B^{\beta_i}$ for every $i\in [n]$.
Computing  $A^{\alpha_i}B^{\beta_i}$  requires 2
exponentiations and hence the ${\sf PRF}_{\rm dlin}$ has closed form efficiency
for $({\sf Comp},z)$, where $z=\{(i,j): i,j\in[n]\}$.
The ${\sf PRF}_{\rm dlin}$ in  \cite{FG12} is pseudorandom merely based on the DLIN  for
bilinear groups.
Similarly,  the generalized ${\sf PRF}_{\rm dlin}$  is also pseudorandom based on the DLIN assumption
for trilinear groups. Consequently, we have the following lemma.
\begin{lemma}
If {DLIN} holds in the  trilinear setting, then ${\sf PRF}_{\rm dlin}$ is
an algebraic PRF with closed form efficiency.
\end{lemma}

\subsection{Verifiable  Computation}
\label{section:vc}

Verifiable computation  \cite{GGP10,BGV11,FG12} is a two-party protocol between  a client
and  a server, where the client  gives  encodings of
its function $f$ and  input $x$ to the server, the server returns an encoding of
$f(x)$ along with a proof, and finally
the client  efficiently  verifies  the
server's computation.
 Formally, a VC scheme $\Pi=
({\sf KeyGen, ProbGen, Compute, Verify})$ is defined by four  algorithms, where
\begin{itemize}
\item $(pk,sk)\leftarrow {\sf KeyGen}(1^\lambda, f)$
takes as input a security parameter $\lambda$ and a function $f$, and generates
both  a public key $pk$ and a secret key $sk$;
\item $(\sigma, \tau)\leftarrow {\sf ProbGen}(sk,x)$
takes as input the secret key $sk$ and an input $x$, and generates both
an encoded input
 $\sigma$  and a verification key $\tau$;
\item $(\rho, \pi)\leftarrow {\sf Compute}(pk,\sigma)$
 takes as input the public key $pk$ and an encoded input $\sigma$, and produces both
an encoded output $\rho$ and a proof $\pi$;
\item $\{f(x),\perp\}\leftarrow {\sf Verify}(sk,\tau,\rho,\pi)$
 takes as input the secret key $sk$, the verification key $\tau$, the
encoded output $\rho$ and a proof $\pi$, and outputs
either $f(x)$  or $\perp$ (which indicates that $\rho$ is not valid).
\end{itemize}

\vspace{1mm}
\noindent
{\bf  Correctness.}
The  scheme $\Pi$ should be correct.
    Intuitively, the scheme $\Pi$ is correct if
     an honest server always  outputs a pair $(\rho,\pi)$
that gives the correct computation result. Let ${\cal F}$ be a family of functions.
\begin{definition}
The  scheme $\Pi$ is  said to be {\em ${\cal F}$-correct} if for any
$f\in {\cal F}$, any $(pk,sk)\leftarrow {\sf KeyGen}(1^\lambda, f)$, any input
$x$  to $f$,
any $(\sigma,\tau)\leftarrow {\sf ProbGen}(sk,x)$,
any $(\rho,\pi)
\leftarrow {\sf Compute}(pk,\sigma)$,
it holds that $f(x)={\sf Verify}(sk,\tau,\rho,\pi)$.
\end{definition}

\vspace{-5mm}
\begin{figure}[h]
\begin{center}
\begin{boxedminipage}{8.5cm}
\begin{enumerate}
\item[] \underline{{\bf Experiment} ${\sf Exp}_{\cal A}^{\sf Ver}(\Pi, f,\lambda)$}
\item $(pk,sk)\leftarrow {\sf KeyGen}(1^\lambda,f)$;
\item for $i=1$ to $l={\sf poly}(\lambda)$ do
\item \hspace{0.5cm}$x_i\leftarrow \mathcal{A}(pk,x_1,\sigma_1,
\ldots, x_{i-1}, \sigma_{i-1})$;
\item \hspace{0.5cm}$(\sigma_i, \tau_i)\leftarrow {\sf ProbGen}(sk,x_i)$;
\item $\hat{x}\leftarrow \mathcal{A}(pk,x_1,\sigma_1, \ldots,x_l,\sigma_l)$
\item $(\hat{\sigma},\hat{\tau})\leftarrow {\sf ProbGen}(sk,\hat{x})$;
\item $(\bar{\rho}, \bar{\pi})\leftarrow \mathcal{A}(pk,x_1,\sigma_1,
\ldots,x_l,\sigma_l,\hat{\sigma})$
\item $\bar{y}\leftarrow {\sf Verify}(sk,\hat{\tau},\bar{\rho}, \bar{\pi})$;
\item output 1 if $\bar{y}\notin \{f(\hat{x}),\perp\}$ and 0 otherwise.
\end{enumerate}
\end{boxedminipage}
\begin{boxedminipage}{8cm}
\begin{enumerate}
\item[] \underline{{\bf Experiment} ${\sf Exp}_{\cal A}^{\sf Pri}(\Pi, f,\lambda)$}
\item $(pk,sk)\leftarrow {\sf KeyGen}(1^\lambda,f)$;
\item $(x_0, x_1)\leftarrow \mathcal{A}^{{\sf PubProbGen}(sk,\cdot)}(pk)$;
\item $b\leftarrow \{0,1\}$;
\item $(\sigma,\tau)\leftarrow {\sf ProbGen}(sk,x_b)$;
\item $b^\prime\leftarrow \mathcal{A}^{{\sf PubProbGen}(sk,\cdot)}(pk, x_0,x_1,\sigma)$
\item output 1 if $b^\prime=b$ and 0 otherwise.
\item[] {\scriptsize Remark:  ${\sf PubProbGen}(sk,\cdot)$
takes as input $x$,
runs $(\sigma,\tau)\leftarrow {\sf ProbGen}(sk,x)$ and returns $\sigma$.}
\vspace{11.25mm}
\end{enumerate}
\end{boxedminipage}
\caption{Experiments for security and privacy \cite{GGP10}}
\label{figure:experiment}
\end{center}
\end{figure}

\vspace{-5mm}
\noindent
{\bf Security.}
The  scheme $\Pi$ should be secure.
As in \cite{GGP10}, we say that the scheme $\Pi$  is secure if no untrusted  server can cause
the client to accept an incorrect computation result with a forged proof.
This intuition can be formalized by an experiment
${\sf Exp}_{\cal A}^{\sf Ver}(\Pi, f,\lambda)$  (see {Fig.} \ref{figure:experiment})
where the challenger  plays the role of
the client and the adversary ${\cal A}$  plays the role of the untrusted  server.
\begin{definition}
 The  scheme $\Pi$ is said to be
{\em ${\cal F}$-secure} if for any $f\in {\cal F}$ and any PPT adversary
${\cal A}$, it holds that
$\Pr[{\sf Exp}_{\cal A}^{\sf Ver}(\Pi, f, \lambda)=1]<{\sf neg}(\lambda).$
\end{definition}

\vspace{1mm}
\noindent
{\bf Privacy.}
The client's input should be hidden from the server in $\Pi$.
 As in \cite{GGP10}, we define  input privacy  based on
the intuition that no untrusted  server can distinguish between different inputs of the client.
This is formalized by an experiment ${\sf Exp}_{\cal A}^{\sf Pri}(\Pi, f,\lambda)$
(see {Fig.} \ref{figure:experiment}) where the challenger  plays the role
of the client and the adversary $\cal A$  plays the role of the untrusted server.
\begin{definition}
\label{definition:privacy}
The  scheme $\Pi$ is said to achieve {\em  input privacy}
 if for any function $f\in {\cal F}$, any PPT algorithm $\cal A$,
it holds that $\Pr[{\sf Exp}_{\cal A}^{\sf Pri}(\Pi, f,\lambda)=1]<{\sf neg}(\lambda).$
\end{definition}

\vspace{1mm}
\noindent
{\bf Efficiency.}
The algorithms ${\sf ProbGen}$ and
  ${\sf Verify}$ will be run by the client for each evaluation of the outsourced function
  $f$. Their running time should be substantially  less than
  evaluating $f$.
\begin{definition}
The  scheme $\Pi$ is said to be {\em outsourced}
if for any $f\in {\cal F}$ and any input $x$ to $f$,
the  running time of  ${\sf ProbGen}$ and ${\sf Verify}$ is $o(t)$, where
$t$ is the time required to compute $f(x)$.
\end{definition}

\section{Our Schemes}
\label{section:schemes}

\subsection{Univariate Polynomial Evaluation}
\label{section:PE}

In this section, we present our VC scheme $\Pi_{\rm pe}$ with input privacy
 (see {Fig.} \ref{figure:vc-pe}) for univariate polynomial evaluation.
In $\Pi_{\rm pe}$, the client outsources  a degree $n$ polynomial
 $f(x)=f_0+f_1x+\cdots+f_nx^n\in \mathbb{Z}_q[x]$  to the server and may evaluate
 $f(\alpha)$ for
 any input $\alpha\in \mathbb{D}\subseteq \mathbb{Z}_q$, where $q$ is a
$\lambda$-bit prime not known to the server and $|\mathbb{D}|={\sf poly}(\lambda)$.
 Our scheme uses a $(2k+1)$-multilinear map instance $\Gamma$ with groups of order $N=pq$,
 where
 $k=\lceil \log (n+1)\rceil$ and $p$ is also a $\lambda$-bit prime not known to the server.
The client stores $t=g_1^{f(s)}$ and  gives $(g_1^s,g_1^{s^2}\ldots, g_1^{s^{2^{k-1}}},f)$  to the server, where $s\leftarrow \mathbb{Z}_N$.
It also sets up ${\rm BGN}_{2k+1}$ based on $\Gamma$.
 In order to verifiably compute
$f(\alpha)$, the client gives $\sigma=(\sigma_1,\ldots, \sigma_k)$ to the server and the server returns $\rho={\sf Enc}(f(\alpha))$ along with
$\pi={\sf Enc}(c(s))$, where $\sigma_\ell={\sf Enc}(\alpha^{2^{\ell-1}})$
for every $\ell\in [k]$ and  $(\rho,\pi)$ is computed using the techniques in Section \ref{section:BGN}.
At last, the client  decrypts $\rho$ to $y$ and verifies if the equation
(\ref{equation:check}) holds.

 \begin{figure}[h]
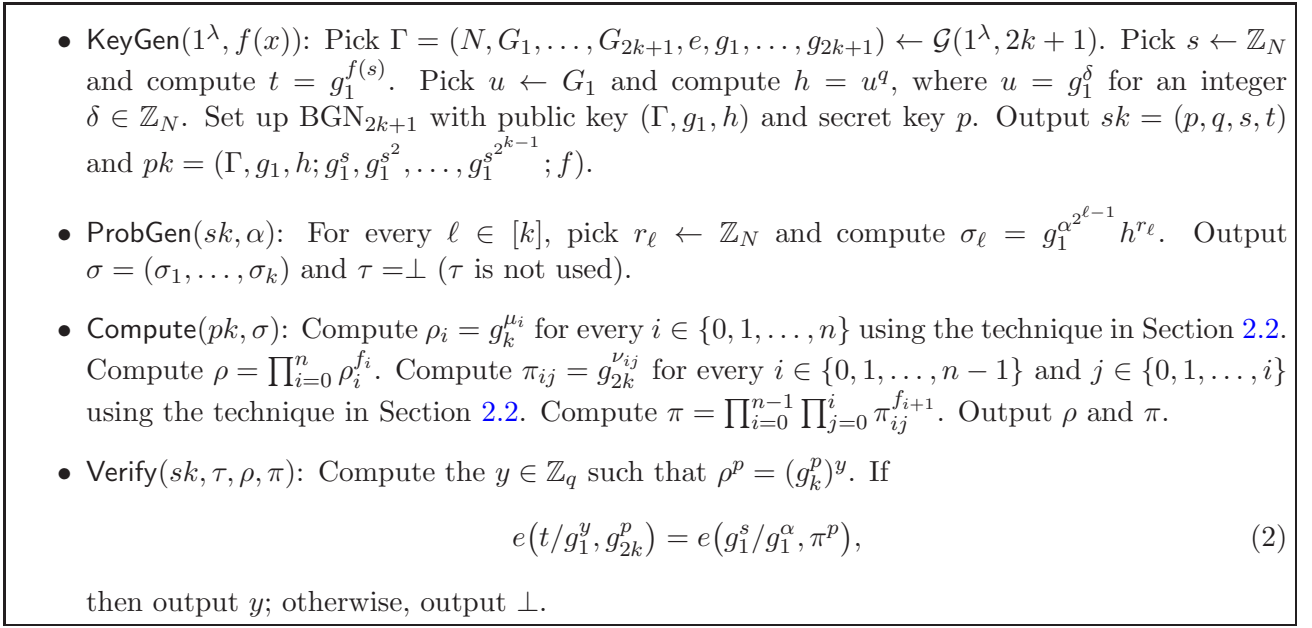

\begin{center}
\begin{boxedminipage}{17cm}
\begin{itemize}
\item  $\mathsf{KeyGen}(1^\lambda, f(x))$:
Pick
 $\Gamma=(N, G_1,\ldots,  G_{2k+1},   e,  g_1, \ldots,  g_{2k+1})\leftarrow {\cal G}(1^\lambda, 2k+1)$.
 Pick  $ s\leftarrow \mathbb{Z}_N$ and compute
 $t=g_1^{f(s)}$.
 Pick
$u\leftarrow G_1$ and compute $h=u^q$, where $u=g_1^{\delta}$ for  an integer
$\delta\in \mathbb{Z}_N$. Set up ${\rm BGN}_{2k+1}$  with public key
$(\Gamma,g_1, h)$ and secret key $p$.
Output $sk=(p,q, s, t)$ and
 $pk=(\Gamma, g_1,h; g_1^s, g_1^{s^2},  \ldots, g_1^{s^{2^{k-1}}}; f)$.
\item ${\sf ProbGen}(sk,\alpha)$:   For every $\ell\in [k]$, pick
 $r_\ell \leftarrow \mathbb{Z}_N$ and compute
 $\sigma_\ell=g_1^{\alpha^{2^{\ell-1}}}h^{r_\ell}$. Output
 $\sigma=(\sigma_1,\ldots, \sigma_k)$ and  $\tau=\perp$
  ($\tau$ is  not used).
\item ${\sf Compute}(pk,\sigma)$: Compute
$\rho_i=g_{k}^{\mu_i}$ for  every   $i\in \{0,1,\ldots,n\}$ using the technique in Section
\ref{section:BGN}. Compute  $\rho=\prod_{i=0}^n \rho_i^{f_i}$.
Compute $\pi_{ij}=g_{2k}^{\nu_{ij}}$ for every $i\in\{0,1,\ldots,n-1\}$
and $j\in\{0,1,\ldots,i\}$ using the technique in Section \ref{section:BGN}.
Compute   $\pi=\prod_{i=0}^{n-1}\prod_{j=0}^i \pi_{ij}^{f_{i+1}}$.
Output  $\rho$ and  $\pi$.
\item ${\sf Verify}(sk, \tau, \rho,\pi)$: Compute  the $y\in \mathbb{Z}_q$ such that
$\rho^p=(g_{k}^p)^y$.
If
\begin{equation}
\label{equation:check}
e\big(t/g_1^{y}  , g_{2k}^{p}\big)=
e\big(g_1^s/g_1^\alpha, \pi^{p}\big),
\end{equation}
 then output $y$; otherwise,
output $\perp$.

\end{itemize}
\end{boxedminipage}
\caption{Univariate polynomial evaluation ($\Pi_{\rm pe}$)}
\label{figure:vc-pe}
\end{center}
\end{figure}

\vspace{1mm}
\noindent
{\bf Correctness.}
The correctness of $\Pi_{\rm pe}$ requires that the client always
outputs $f(\alpha)$ as long as the server is honest, i.e.,
$y=f(\alpha)$ and (\ref{equation:check}) holds.
It is shown by the following lemma.
\begin{lemma}
\label{theorem:pe-correctness}
If the server is honest, then
$y=f(\alpha)$ and
  (\ref{equation:check}) holds.
\end{lemma}

\begin{proof}
Firstly,   $p\mu_i\equiv p\alpha^i\bmod N$ for every
$i\in \{0,1,\ldots,n\}$.
Since $g_{k}$ is of
order $N$,
we have that
$$\rho^p=\prod_{i=0}^n \rho_i^{pf_i}=\prod_{i=0}^n g_{k}^{p\mu_if_i}
=\prod_{i=0}^n g_{k}^{pf_i\alpha^i}
=\big(g_{k}^p\big)^{f(\alpha)},
$$
which implies that $y=f(\alpha)$.
 Secondly, $p\nu_{ij}=p\alpha^j s^{i-j}\bmod N$ for every $i\in \{0,1,\ldots, n-1\}$
 and $j\in\{0,1,\ldots,i\}$.
Thus,
$$\pi^p=\prod_{i=0}^{n-1}\prod_{j=0}^i \pi_{ij}^{p f_{i+1}}=
\prod_{i=0}^{n-1}\prod_{j=0}^i  g_{2k}^{p \nu_{ij} f_{i+1}}=
\prod_{i=0}^{n-1}\prod_{j=0}^i g_{2k}^{pf_{i+1}\alpha^j s^{i-j}}=g_{2k}^{c(s)}. $$
It follows that
$e\big(g_1^s/g_1^\alpha, \pi^{p}\big)=g_{2k+1}^{(s-\alpha)c(s)}=g_{2k+1}^{f(s)-f(\alpha)}=
e\big(t/{g_1^y},
 g_{2k}^{p}\big),$ i.e, the equality (\ref{equation:check}) holds.
\end{proof}

\newpage

\noindent
{\bf Security.}
The security  of $\Pi_{\rm pe}$ requires that no untrusted server can cause
 the client to accept a value $\bar{y}\neq f(\alpha)$ with a forged proof.
 It is based on  the $(2k+1, n)$-MSDHS  assumption (see Definition \ref{definition:msdh}).
\begin{lemma}
\label{theorem:security-pe}
If
$(2k+1, n)$-{MSDHS} holds for $\Gamma$, then the scheme
$\Pi_{\rm pe}$ is secure.
\end{lemma}
\begin{proof}
Suppose that $\Pi_{\rm pe}$ is not secure. Then there is
a PPT adversary  ${\cal A}$ that breaks its
security with non-negligible probability $\epsilon_1$. We shall construct a  PPT simulator ${\cal B}$ that
simulates ${\cal A}$ and breaks the $(2k+1,n)$-MSDHS for $\Gamma$. The simulator
${\cal B}$ takes as input
$(p,q,\Gamma,g_1, g_1^s, \ldots, g_1^{s^{n}})$, where  $s\leftarrow \mathbb{Z}_N$.
The simulator ${\cal B}$ is required to output  $g_{2k+1}^{p/s}$. In order to do so, ${\cal B}$ simulates
${\cal A}$ as  below:
\begin{itemize}
\item[(A)] Pick  a polynomial $f(x)=f_0+f_1x+\cdots+f_n x^n\in \mathbb{Z}_q[x]$.
Pick $u\leftarrow G_1$, compute $h=u^q$ and set up ${\rm BGN}_{2k+1}$
with public key $(\Gamma,g_1,h)$ and secret key $p$. Pick $\beta\leftarrow \mathbb{D}$ and implicitly set $\hat{s}=s+\beta$ ($\hat{s}$ is not known to $\cal B$).
Mimic
${\sf KeyGen}$ by sending   $pk=(\Gamma, g_1,h,  g_1^{\hat{s}}, g_1^{\hat{s}^2}\ldots,  g_1^{\hat{s}^{2^{k-1}}},f)$
 to ${\cal A}$ (note that $\cal B$ can compute $g_1^{\hat{s}^{2^{\ell-1}}}$ for every $\ell\in [k]$ based on the knowledge of $\beta$ and $g_1,g_1^s,\ldots, g_1^{s^n}$). Set $sk=(p,q,t)$, where $t=g_1^{f(\hat{s})}$ (note that $sk$ does not
 include $\hat{s}$ as a component because $\hat{s}$ is neither known to $\cal B$ nor used by $\cal B$);
\item[(B)]  Upon receiving  $\alpha\in \mathbb{D}$ from ${\cal A}$, mimic ${\sf ProbGen}$
as below:
 pick $r_\ell\leftarrow \mathbb{Z}_N$ and compute
$\sigma_\ell=g_1^{\alpha^{2^{\ell-1}}}h^{r_\ell}$ for every $\ell\in [k]$;
send $\sigma=(\sigma_1,\ldots, \sigma_k)$ to  ${\cal A}$.
\end{itemize}
It is trivial to verify that the $pk$ and $\sigma$ generated by ${\cal B}$
 are identically distributed to those generated by the client in an execution
of  $\Pi_{\rm pe}$.
We remark that (A) is the step 1 in  ${\sf Exp}_{\cal A}^{\sf Ver}(\Pi, f,\lambda)$
(see {Fig.} \ref{figure:experiment}) and
(B) consists of steps 3 and 4 in  ${\sf Exp}_{\cal A}^{\sf Ver}(\Pi, f,\lambda)$.
Furthermore, (B) may be run $l={\sf poly}(\lambda)$  times as
described by step 2 of  ${\sf Exp}_{\cal A}^{\sf Ver}(\Pi, f,\lambda)$.
After $l$  executions of (B), the adversary ${\cal A}$
will provide an input $\hat{\alpha}$ on which he is willing
to be challenged.
If $\hat{\alpha}\neq \beta$, then the simulator $\cal B$ aborts; otherwise, it continues.
Note that both $\beta$ and $\hat{\alpha}$ are from the same polynomial size
domain $\mathbb{D}$, the event that $\hat{\alpha}=\beta$ will occur with probability $\epsilon_2\geq 1/|\mathbb{D}|$, which is non-negligible.
If the simulator ${\cal B}$ does not abort, it next
runs $(\hat{\sigma}, \hat{\tau})\leftarrow {\sf ProbGen}(sk,\hat{\alpha})$  and
gives ${\cal A}$  an encoded input $\hat{\sigma}$. Then  the adversary
${\cal A}$ may  maliciously reply with $(\bar{\rho}, \bar{\pi})$
such that
${\sf Verify}(sk,\hat{\tau}, \bar{\rho}, \bar{\pi})\triangleq \bar{y}\notin \{f(\hat{\alpha}),\perp\}$.
On the other hand,  an honest
server in  $\Pi_{\rm pe}$ will reply with $(\hat{\rho},\hat{\pi})$. Due to Theorem  \ref{theorem:pe-correctness},
it must be the case that
${\sf Verify}(sk, \hat{\tau}, \hat{\rho},\hat{\pi})\triangleq \hat{y}=f(\hat{\alpha})$. Note that the event that
$\bar{y}\notin \{f(\hat{\alpha}), \perp\}$ occurs with probability $\epsilon_1$.
Suppose the event $\bar{y}\notin \{f(\hat{\alpha}), \perp\}$ occurs, then  the equation
(\ref{equation:check}) is satisfied by both $(\bar{y}, \bar{\pi})$
and $(\hat{y},\hat{\pi})$, i.e.,
\begin{equation}
\label{equation:compare-pe}
\begin{split}
e\big(t/g_1^{\bar{y}}, g_{2k}^{p}\big)=
e\big(g_1^{\hat{s}}/g_1^{\hat{\alpha}}, \bar{\pi}^{p}\big)
{\hspace{1mm}}{\rm and} {\hspace {1mm}}
e\big(t/g_1^{\hat{y}}, g_{2k}^{p}\big)=
e\big(g_1^{\hat{s}}/g_1^{\hat{\alpha}}, \hat{\pi}^{p}\big).
\end{split}
\end{equation}
The equalities in (\ref{equation:compare-pe}) imply  that
$\displaystyle e\big(g_1^{\bar{y}-\hat{y}}, g_{2k}^{p}\big)=e\big(g_1^{\hat{s}-
\hat{\alpha}},
\big(\hat{\pi}/\bar{\pi}\big)^{p}\big).$
Hence,
\begin{equation}
\label{equation:break}
g_{2k+1}^{\frac{p}{\hat{s}-\hat{\alpha}}}=e\big(g_1,
\big(\hat{\pi}/\bar{\pi}\big)^{p}\big)^{\frac{1}{\bar{y}-\hat{y}}}.
\end{equation}
Note that the left hand side of (\ref{equation:break}) is $g_{2k+1}^{p/s}$ due to $\beta=\hat{\alpha}$.
Therefore, (\ref{equation:break}) means that
the simulator ${\cal B}$ can break the $(2k+1,n)$-MSDHS assumption (Definition \ref{definition:msdhs}) with probability $\epsilon=\epsilon_1\epsilon_2$, which is non-negligible and  contradicts to the $(2k+1,n)$-MSDHS assumption.
Hence, under the $(2k+1,n)$-MSDHS assumption,
$\epsilon_1$ must be negligible in $\lambda$, i.e., the scheme $\Pi_{\rm pe}$ is secure.
\end{proof}

\noindent
{\bf Privacy.}
The input privacy of  $\Pi_{\rm pe}$ requires
that no untrusted server
can distinguish between  different inputs  of the client. This is formally defined by
the experiment ${\sf Exp}_{\cal A}^{\sf Pri}(\Pi, f,\lambda)$ in
 {Fig.} \ref{figure:experiment}. The client in our VC scheme encrypts its input
 $\alpha$ using ${\rm BGN}_{2k+1}$ which is semantically secure under
 SDA for $\Gamma$. As a result, our VC scheme achieves input privacy under SDA for
 $\Gamma$ (see Appendix \ref{appendix:privacy-pe} for a proof of the following lemma).
\begin{lemma}
\label{theorem:privacy-PE}
If {SDA} holds for $\Gamma$, then the scheme
$\Pi_{\rm pe}$ achieves the input privacy.
\end{lemma}

\noindent
{\bf Efficiency.}
In order to verifiably compute $f(\alpha)$ with the cloud, the client
computes $k=\lceil \log (n+1) \rceil$ ciphertexts $\sigma_1,\ldots, \sigma_k$
under ${\rm BGN}_{2k+1}$
in the execution of $\sf ProbGen$; it also  decrypts one ciphertext
$\rho={\sf Enc}(f(\alpha))$
under ${\rm BGN}_{2k+1}$ and then verifies the equation (\ref{equation:check}).
 The overall computation of the client
 will be $O(\log n)=o(n)$ and therefore $\Pi_{\rm pe}$ is outsourced.
On the other hand, the server needs to perform $O(n^2\log n)$
multilinear map computations and $O(n^2)$ exponentiations  in each execution of
$\sf Compute$, which is comparable with the VC schemes based on FHE.
Based on Lemmas \ref{theorem:pe-correctness},  \ref{theorem:security-pe},
 \ref{theorem:privacy-PE} and the  efficiency analysis,  we have the following theorem.
\begin{theorem}
If the $(2k+1, n)$-{MSDHS} and SDA assumptions for $\Gamma$ both hold, then
 $\Pi_{\rm pe}$ is a VC scheme with input privacy.
\end{theorem}

\vspace{1mm}
\noindent
{\bf A variant of $\Pi_{\rm pe}$ based on $(2k+1,n)$-{\rm MSDH}.}
The VC scheme $\Pi_{\rm pe}$ we constructed in this section has its security based on
the $(2k+1,n)$-MSDHS which is slightly stronger than the
 $(2k+1,n)$-MSDH.
The $(2k+1,n)$-MSDH  implies that
except for a negligible fraction of the $(2k+1)$-multilinear map instances,
at least one of the following problems  is difficult for any PPT algorithm:
(i) given $g_1,g_1^s,\ldots, g_1^{s^n}$, compute $g_1^{p/s}$;
(ii) given $g_1,g_1^s,\ldots, g_1^{s^n}$, compute $g_1^{q/s}$.
While it is not known how to determine which one of the two problems  is difficult
for a given $\Gamma_{2k+1}$ instance, the construction  of $\Pi_{\rm pe}$ simply assumes one of them
is always difficult. In fact, we can also construct a VC scheme with input privacy
for univariate polynomial evaluation based on the weaker assumption
$(2k+1,n)$-MSDH. A natural idea is generating multiple $(2k+1)$-multilinear map instances,
say $\Gamma_{2k+1,1}, \ldots, \Gamma_{2k+1,\lambda}$, where
$\Gamma_{2k+1,l}$ is defined over groups of order $N_l=p_lq_l$ for every $l\in [\lambda]$.
The client can simply run one instance of
$\Pi_{\rm pe}$ based on each one of the $\lambda$ multilinear map instances
$\Gamma_{2k+1,1}, \ldots, \Gamma_{2k+1,\lambda}$. In particular, in each execution,
the client simply  picks a random one from the two problems and believe that the chosen
problem  is difficult. It will base that execution on the hardness of the chosen problem.
The client will not accept the $\lambda$ computation results from
the $\lambda$ executions of $\Pi_{\rm pe}$ except all of them agree with each other and
passes their respective verifications. The client will output a wrong computation result
only when the $\lambda$  chosen problems are
all easy and the server turns out to be able to break all of them. However, this will happen
with probability at most $2^{-\lambda}$, which is negligible.
Therefore, this modified scheme will be secure merely based on
 $(2k+1,n)$-MSDH.  Surely, as the price of using a weaker assumption,
the resulting scheme incurs a $\lambda$ overhead  to the  computation and communication complexities
of every party in $\Pi_{\rm pe}$,
which however is acceptable in some scenarios because the
$\lambda$ is independent of $n$.

\subsection{Matrix Multiplication}
\label{section:MM}

In this section, we present our VC scheme $\Pi_{\rm mm}$ with input privacy
(see {Fig.} \ref{figure:vc-mm})  for matrix multiplication.
In $\Pi_{\rm mm}$, the client outsources an $n\times n$ matrix
 $M=(M_{ij})$ over $\mathbb{Z}_q$ to the server and
may  compute $Mx$ for an input vector  $x=(x_1,\ldots,x_n)\in \mathbb{D}\subseteq \mathbb{Z}_q^n$,
where $q$ is a $\lambda$-bit prime not known to the server and $|\mathbb{D}|={\sf poly}(\lambda)$.
Our scheme uses a  trilinear  map instance $\Gamma$ with groups of order $N=pq$,
where $p$ is also a $\lambda$-bit
 prime not known to the server.
 In  $\Pi_{\rm mm}$, the client gives both $M$ and its
blinded version $T=(T_{ij})$ to the server, where $T$ is computed using
the ${\sf PRF}_{\rm dlin}$.
It also sets up ${\rm BGN}_3$.
In order to verifiably compute $Mx$,
the client stores $\tau=(\tau_1,\ldots, \tau_n)$, where each
  $\tau_i$ is efficiently computed using the
 closed form efficiency property of ${\sf PRF}_{\rm dlin}$.
 It gives  $\sigma=({\sf Enc}(x_1),\ldots, {\sf Enc}(x_n))$
  to  the server
 and the server   returns    $\rho=(\rho_1,\ldots, \rho_n)={\sf Enc}(Mx)$  along with
$\pi=(\pi_1,\ldots, \pi_n)$.
At last, the client  decrypts $\rho_i$ to $y_i$
and verify if (\ref{equation:check-MM}) holds for every $i\in[n]$.

%
\begin{figure}[h]
\begin{center}
\begin{boxedminipage}{17cm}
\begin{itemize}
\item ${\sf KeyGen}(1^\lambda,M)$:
Pick a trilinear map instance $\Gamma=(N,G_1,G_2,G_3,e,g_1,g_2,g_3)\leftarrow {\cal G}(1^\lambda,3)$.
 Consider the ${\sf PRF}_{\rm dlin}$ in Section \ref{section:aprf}.
Run  ${\sf KG}(1^\lambda,n)$
and pick a  secret key $K$.
Pick  $a\leftarrow \mathbb{Z}_N$ and compute
 $T_{ij}=g_1^{p^2aM_{ij}}\cdot {\sf   F}_K(i,j)$ for every  $(i,j)\in[n]^2$.
 Pick  $u\leftarrow G_1$ and compute
$ h=u^q$.
Set up ${\rm BGN}_{3}$ with public key
$(\Gamma,g_1,h)$ and secret key $p$.
Output
 $sk=(p,q,K, a, \eta)$ and  $pk=(\Gamma, g_1,h, M,T)$, where $\eta=g_3^{p^2a}$.

\item ${\sf ProbGen}(sk,x)$:
For every $j\in[n]$, pick  $r_j\leftarrow \mathbb{Z}_N$
and compute  $\sigma_j=g_1^{x_j}h^{r_j}$.
For every $i\in[n]$, compute
 $\tau_i= e(\prod_{j=1}^n{\sf F}_K(i,j)^{x_j}, g_2^{p})$ using the efficient ${\sf CFE}$
 algorithm in Section \ref{section:aprf}.
Output   $\sigma=(\sigma_1,\ldots,\sigma_n)$ and
 $\tau=(\tau_1,\ldots, \tau_n)$.
\item ${\sf Compute}(pk,\sigma)$:
Compute
 $\rho_i=\prod_{j=1}^n \sigma_j^{M_{ij}}$ and
 $\pi_i=\prod_{j=1}^n e(T_{ij},\sigma_j)$ for every $i\in[n]$. Output
  $\rho=(\rho_1,\ldots,\rho_n)$ and
 $\pi=(\pi_1,\ldots, \pi_n)$.
\item ${\sf Verify}(sk,\tau, \rho,\pi)$:
For every $i\in[n]$, compute  $y_i$ such that $\rho_i^p=(g_1^p)^{y_i}$.
If
\begin{equation}
\label{equation:check-MM}
e(\pi_i, g_1^p)=\eta^{py_i}\cdot \tau_i
\end{equation}
 for every $i\in[n]$, then output $y=(y_1,\ldots,y_n)$;
otherwise output $\perp$.
\end{itemize}
\end{boxedminipage}
\caption{Matrix multiplication ($\Pi_{\rm mm}$)}
\label{figure:vc-mm}
\label{figure:vcmm}
\end{center}
\end{figure}

\vspace{1mm}
\noindent
{\bf Correctness.}
The correctness of $\Pi_{\rm mm}$ requires that the client  always outputs
$Mx$ as long as the server is honest, i.e., $y=Mx$
and  (\ref{equation:check-MM}) holds for every $i\in[n]$.
It is shown by the following lemma.
\begin{lemma}
\label{theorem:MM-correctness}
If the server is honest, then
$y=Mx$ and  (\ref{equation:check-MM}) holds
for every $i\in[n]$.
\end{lemma}
\begin{proof}
Firstly,   $\sigma_j^p=g_1^{px_j}$ for every $j\in [n]$. Thus we have that
$$
\rho_i^p=\prod_{j=1}^n \sigma_j^{pM_{ij}}
=\prod_{j=1}^n g_1^{pM_{ij}x_j}=\big(g_1^p\big)^{\sum_{j=1}^n M_{ij}x_j}
$$
for every $i\in[n]$.
It follows that $y_i=\sum_{j=1}^n M_{ij}x_j\in \mathbb{Z}_q$
for every $i\in[n]$. Hence, $y=Mx$.   Secondly, we have that
$$
e(\pi_i, g_1^p)=
\prod_{j=1}^n e(e(T_{ij}, \sigma_j), g_1^p)=
\prod_{j=1}^n e(g_2^{p^2aM_{ij}x_j}, g_1^p)\cdot
 e(\prod_{j=1}^n {\sf F}_K(i,j)^{x_j}, g_2^{p})=\eta^{py_i}\cdot \tau_i,
$$
i.e., the equality (\ref{equation:check-MM}) holds.
\end{proof}

\noindent
{\bf Security.}
The security of $\Pi_{\rm mm}$ requires that no untrusted server
can cause the client to accept $\bar{y}
\notin\{ Mx, \perp\}$ with a forged proof.
It is based on the 3-co-CDHS assumption   for $\Gamma$
(Lemma \ref{lemma:3-co-cdhs})
 and the DLIN assumption (Definition \ref{definition:dlin}).
\begin{lemma}
\label{theorem:security-mm}
If  the {3-co-CDHS} assumption for $\Gamma$ and the {DLIN} assumption both hold,
then the scheme
$\Pi_{\rm mm}$ is secure.
\end{lemma}

\begin{proof}
We define  three games ${\sf G}_0,{\sf G}_1$ and ${\sf G}_2$ as below:
\begin{itemize}
\item[${\sf G}_0:$] this  is the standard security game
              ${\sf Exp}_{\cal A}^{\sf Ver}(\Pi,M,\lambda)$ defined in
              {Fig.} \ref{figure:experiment}.
\item[${\sf G}_1:$] the only difference between this game and ${\sf G}_0$
is a change to  {\sf ProbGen}. For any $(x_1,\ldots,x_n)$ queried  by the adversary, instead of computing
$\tau$ using the efficient {\sf CFE} algorithm, the inefficient evaluation of $\tau_i$ is used,
i.e., $\tau_i=\prod_{j=1}^n e({\sf F}_K(i,j)^{x_j}, g_2^{p})$ for every $i\in[n]$.
\item[${\sf G}_2:$] the only difference between this game and ${\sf G}_1$ is that
the matrix $T$ is computed as $T_{ij}=g_1^{p^2aM_{ij}}\cdot R_{ij}$,
where $R_{ij}\leftarrow G_1$  for every $i,j\in[n]$.
\end{itemize}
For every $i\in \{0,1,2\}$, we denote by
${\sf G}_i({\cal A})$ the output of game $i$ when it is run with an adversary ${\cal A}$.
The proof of the theorem proceeds by a standard hybrid argument, and is obtained by combining
the proofs of the following three claims.

\vspace{1mm}
\noindent\underline{\rm Claim 1}. We have that $\Pr[{\sf G}_0({\cal A})=1]=
\Pr[{\sf G}_1({\cal A})=1]$.

The only difference between ${\sf G}_1$ and  ${\sf G_0}$
is in the
computation of $\tau$.   Due to the correctness of the ${\sf CFE}$
algorithm,  such difference does not change the distribution of the
values  $\tau$ returned to the adversary. Therefore,
 the probabilities that ${\cal A}$
wins in both games are identical.

\vspace{0.1cm}
\noindent\underline{\rm Claim 2}. We have that $|\Pr[{\sf G}_1({\cal A})=1]-
\Pr[{\sf G}_2({\cal A})=1]|<{\sf neg}(\lambda)$.

The only difference between  ${\sf G}_2$ and ${\sf G_1}$ is
 that we replace the pseudorandom group elements  ${\sf F}_K(i,j)$ with truly random
 group elements $R_{ij}\leftarrow G_1$ for every $i,j\in [n]$.
Clearly, if $|\Pr[{\sf G}_1({\cal A})=1]-
\Pr[{\sf G}_2({\cal A})=1]|$ is non-negligible, we can construct an simulator ${\cal B}$ that
simulates  ${\cal A}$ and breaks the pseudorandom property of  $\sf PRF$
with a non-negligible advantage.

\vspace{0.1cm}
\noindent\underline{\rm Claim 3}. We have that $\Pr[{\sf G}_2({\cal A})=1]<
{\sf neg}(\lambda)$.

 Suppose that there is  a  PPT adversary
${\cal A}$ that wins with   non-negligible probability $\epsilon$ in  ${\sf G}_2$.
We want to construct a PPT simulator ${\cal B}$ that simulates
${\cal A}$ and breaks the 3-co-CDHS assumption
 (see Definition \ref{definition:3-co-cdhs}) with  non-negligible probability.
The adversary ${\cal B}$ takes as input a tuple
 $(p,q,\Gamma,h_1^\alpha, h_2^\beta)$, where $h_1=g_1^p, h_2=g_2^p$ and $\alpha,\beta\leftarrow \mathbb{Z}_N$.
The adversary   ${\cal B}$ is required to output
$h_2^{\alpha\beta}$. In order to do so, ${\cal B}$ simulates ${\cal A}$ as below:
\begin{enumerate}
\item[(A)] Pick an  $n\times n$ matrix $M$ and mimic the ${\sf KeyGen}$
 of game ${\sf G}_2$ as below:
\begin{itemize}
\item implicitly set $a=\alpha\beta$ by computing
$\eta=e(h_1^{\alpha}, h_2^{\beta})=g_3^{p^2\alpha\beta}$;
\item pick  $u\leftarrow G_1$, compute $h=u^q$ and set up
${\rm BGN}_3$ with public key $(\Gamma,g_1,h)$ and secret key $p$;
\item pick  $T_{ij}
\leftarrow G_1$ for every $i,j\in [n]$ and send $pk=(\Gamma,g_1,h,M,T)$ to ${\cal A}$, where   $T=(T_{ij})$;
\end{itemize}
\item[(B)] Upon receiving a query $x=(x_1,\ldots, x_n)$ from   ${\cal A}$,
mimic ${\sf ProbGen}$ as below:
\begin{itemize}
\item  for every $j\in [n]$, pick  $r_j\leftarrow \mathbb{Z}_N$ and
compute  $\sigma_j=g_1^{x_j}h^{r_j}$;
\item for every $i,j\in [n]$, compute
$Z_{ij}=e(T_{ij},g_2^{px_j})/\eta^{pM_{ij}x_j}$;
\item for every $i\in[n]$, compute $\tau_i=\prod_{j=1}^n Z_{ij}$;
\item send $\sigma=(\sigma_1,\ldots, \sigma_n)$ to ${\cal A}$.
\end{itemize}
\end{enumerate}
It is straightforward to verify that the $pk, \sigma$ and $\tau$ generated
by ${\cal B}$  are identically
distributed to those generated by the client in game ${\sf G}_2$.
We remark that (A) is the step 1 in  ${\sf Exp}_{\cal A}^{\sf Ver}(\Pi, M,\lambda)$
(see {Fig.} \ref{figure:experiment}) and
(B) consists of steps 3 and 4 in  ${\sf Exp}_{\cal A}^{\sf Ver}(\Pi, M,\lambda)$.
Furthermore, (B) may be run $l={\sf poly}(\lambda)$  times as
described by step 2 of  ${\sf Exp}_{\cal A}^{\sf Ver}(\Pi, M,\lambda)$.
After $l$  executions of (B),
 the adversary ${\cal A}$ will provide an input
 $\hat{x}=(\hat{x}_1,\ldots, \hat{x}_n)$ on which he is willing to be challenged.
Upon receiving  $\hat{x}$, the simulator ${\cal B}$ mimics
 ${\sf ProbGen}$ as (B) and gives
${\cal A}$ an encoded input $\hat{\sigma}$. Then the adversary ${\cal A}$
may maliciously reply with
 $\bar{\rho}=(\bar{\rho}_1,\ldots, \bar{\rho}_n)$ and
$\bar{\pi}=(\bar{\pi}_1,\ldots, \bar{\pi}_n)$ such that
${\sf Verify}(sk,\hat{\tau},\bar{\rho}, \bar{\pi})\triangleq\bar{y}\notin \{
M\hat{x}, \perp\}$.
On the other hand,
an honest server in our VC scheme will reply with
 $\hat{\rho}=\
(\hat{\rho}_1,\ldots, \hat{\rho}_n)$
and $\hat{\pi}=(\hat{\pi}_1,\ldots, \hat{\pi}_n)$.
Due to  Lemma \ref{theorem:MM-correctness}, it must be the case that
${\sf Verify}(sk,\hat{\tau}, \hat{\rho},\hat{\pi})\triangleq \hat{y}=M\hat{x}$.
Note that the event
$\bar{y}\notin \{M\hat{x},\perp\}$ occurs with probability $\epsilon$.
Suppose it occurs. Then there
is an integer $i\in[n]$ such that
$\bar{y}_i\neq \hat{y}_i$.
Note that neither $\bar{y}$ nor $\hat{y}$ is $\perp$,
the equation (\ref{equation:check-MM}) must be satisfied by both
$(\bar{y}, \bar{\pi})$ and
$(\hat{y}, \hat{\pi})$, which translates into
 $
      e(\bar{\pi}_i, g_1^p)=\eta^{p\bar{y}_i}\cdot \hat{\tau}_i
$ and $
e(\hat{\pi}_i, g_1^p)=\eta^{p\hat{y}_i}\cdot \hat{\tau}_i$, we have that
$$e(\hat{\pi}_i/\bar{\pi}_i, g_1^p)=
\eta^{p(\hat{y}_i-\bar{y}_i)}=e(g_2^{p^2\alpha\beta(\hat{y}_i-\bar{y}_i)}, g_1^p),
$$
which in turn implies that
$\hat{\pi}_i/\bar{\pi}_i=g_2^{p\alpha\beta \cdot p(\hat{y}_i-\bar{y}_i)}$.
Let $\phi\in \mathbb{Z}_q^*$ be the multiplicative inverse of
$p(\hat{y}_i-\bar{y}_i)\in \mathbb{Z}_q^*$. Then
$g_2^{p\alpha\beta}=(\hat{\pi}_i/\bar{\pi}_i)^{\phi}$, i.e.,
$h_2^{\alpha\beta}=(\hat{\pi}_i/\bar{\pi}_i)^\phi$, which implies that
 ${\cal B}$ can break the 3-co-CDHS with probability at least $\epsilon$.
Therefore, this $\epsilon$ must be negligible in $\lambda$, i.e.,  $\Pr[{\sf G}_2({\cal A})=1]<
{\sf neg}(\lambda)$.
\end{proof}

\noindent
{\bf Privacy.}
The input privacy of  $\Pi_{\rm mm}$ requires that no untrusted server
can distinguish between different  inputs of the client. This is formally
defined by the
experiment ${\sf Exp}_{\cal A}^{\sf Pri}(\Pi, f,\lambda)$ in {Fig.} \ref{figure:experiment}.
The client in our VC scheme encrypts its input $x$ using ${\rm BGN}_{3}$
which is semantically secure under SDA for $\Gamma$. As a result,  $\Pi_{\rm mm}$
 achieves input privacy under SDA for $\Gamma$ (see
 Appendix \ref{appendix:privacy-MM} for proof).

\begin{lemma}
\label{theorem:privacy-MM}
If  the {SDA}  for $\Gamma$ holds,
then  $\Pi_{\rm mm}$
achieves the input privacy.
\end{lemma}

\noindent
{\bf Efficiency.}
In order to verifiably compute $Mx$ with the cloud, the client
computes $n$ ciphertexts $\sigma_1,\ldots, \sigma_k$  under ${\rm BGN}_3$
and $n$ verification keys $\tau_1,\ldots, \tau_n$
in the execution of $\sf ProbGen$; it also   decrypts $n$ ciphertext
$\rho={\sf Enc}(Mx)$
under ${\rm BGN}_3$ and then verifies the equation (\ref{equation:check-MM}).
The overall computation of the client
 will be $O(n)=o(n^2)$ and therefore $\Pi_{\rm pe}$ is outsourced.
 On the other hand, the server needs to perform $O(n^2)$
multilinear map computations and $O(n)$ exponentiations in each execution of
$\sf Compute$, which is comparable with the VC schemes based on FHE.
Based on Lemmas \ref{theorem:MM-correctness},
\ref{theorem:security-mm},  \ref{theorem:privacy-MM}
and the  efficiency analysis,  we have the following theorem.
\begin{theorem}
If  the {3-co-CDHS},  {DLIN} and  {SDA} assumptions  for $\Gamma$ all hold,
then  $\Pi_{\rm mm}$ is a VC scheme with input privacy.
\end{theorem}

\subsection{Discussions}
\label{section:discussion}

A theoretical limitation of our VC schemes $\Pi_{\rm pe}$
and $\Pi_{\rm mm}$ is that the computation results (i.e., $f(\alpha)$ and
$Mx$) must belong to a polynomial size domain
$\mathbb{M}$ since otherwise the client will not be able to decrypt $\rho$ and
then verify its correctness. However, we stress that this is not a real limitation when we
apply both schemes in outsourcing PIR (see Section \ref{section:applications})
where the computation results are either 0 or 1.
On the other hand,  with $f(x)$ and the knowledge ``$f(\alpha)\in \mathbb{M}$"
(resp. $M$ and  the knowledge
``$Mx\subseteq \mathbb{M}$"), one may argue that
the cloud can also learn a polynomial size domain $\mathbb{D}$ where $\alpha$ (resp. $Mx$) is
drawn from and therefore guess the actual value of $\alpha$ (resp.
 $x$) with non-negligible
probability. However, recall  that our privacy experiment
${\sf Exp}_{\cal A}^{\sf Pri}(\Pi, f,\lambda)$ in {Fig.} \ref{figure:experiment}
only requires   the indistinguishability of different inputs. This is achieved
 by $\Pi_{\rm pe}$ and $\Pi_{\rm mm}$ (though for polynomial size domains)
and suffices for our applications.
Furthermore,
in Section \ref{section:ext},
we shall show how to modify $\Pi_{\rm pe}$  and $\Pi_{\rm mm}$
 such that the functions (i.e., $f(x)$ and $M$) are encrypted and then given to the cloud.
 As a consequence, the cloud learns  no information on
 either the outsourced function or input unless it can break the underlying
encryption scheme.

\subsection{Function Privacy}
\label{section:ext}
Note that  $\Pi_{\rm pe}$ and $\Pi_{\rm mm}$
only achieve input privacy.  We say that a VC scheme achieves {\em function privacy} if
the server cannot learn any information about the outsourced  function.
A formal definition of function privacy can be given using an experiment
similar to ${\sf Exp}_{\cal A}^{\sf Pri}(\Pi, f,\lambda)$.
Both
$\Pi_{\rm pe}$ and $\Pi_{\rm mm}$ can be modified such that
function privacy is also achieved. In the modified VC scheme
$\Pi_{\rm pe}^\prime$  (see {Fig.} \ref{figure:pp-vc-pe}
in Appendix \ref{appendix:modified-vc}), the client gives
${\rm BGN}_{2k+2}$ ciphertexts ${\sf Enc}(f)=({\sf Enc}(f_0),\ldots, {\sf Enc}(f_n))$
and $\sigma=({\sf Enc}(\alpha),\ldots, {\sf Enc}(\alpha^{2^{k-1}}))$ to the server.
Then the server can compute  $\rho={\sf Enc}(f(\alpha))$ along with a proof $\pi={\sf Enc}(c(s))$ using
${\sf Enc}(f)$ and
$\sigma$.
In the modified VC scheme  $\Pi_{\rm mm}^\prime$ (see {Fig.} \ref{figure:pp-vc-mm} in Appendix
\ref{appendix:modified-vc}), the client gives  ${\rm BGN}_{3}$  ciphertexts
${\sf Enc}(M)=({\sf Enc}(M_{ij}))$ and $\sigma=({\sf Enc}(x_1),\ldots, {\sf Enc}(x_n))$ to the server.
Then  the server can compute ${\sf Enc}(\sum_{j=1}^n M_{ij}x_j)$ along with a proof
$\pi_i$ using
${\sf Enc}(M)$ and  $\sigma$ for every $i\in[n]$.
It is not hard to prove that the schemes $\Pi_{\rm pe}^\prime$
and $\Pi_{\rm mm}^\prime$ are secure and achieve both input and  function privacy.

\section{Applications}
\label{section:applications}

Our VC schemes  have application in  outsourcing
{private information retrieval} (PIR).
PIR  \cite{KO97} allows a client  to retrieve any bit $w_i$
of a database $w=w_1\cdots w_n\in \{0,1\}^n$ from a  remote server  without revealing
$i$ to the server.
In a trivial solution of PIR, the client simply downloads  $w$
and extracts $w_i$.
The main drawback of this solution is its prohibitive communication cost (i.e. $n$).
In  \cite{KO97,CMS99,GR05}, PIR schemes with non-trivial communication complexity
$o(n)$ have been constructed  based on various  cryptographic assumptions.
However, all of them  assume that the server is
{\em honest-but-curious}.
In real-life scenarios, the server may have
 strong incentive to give the client  an incorrect response.
 Such malicious behaviors may cause the client to make completely wrong
 decisions in its economic activities (say the client is retrieving price information
 from a stock database and deciding in which stock it is going to invest).
Therefore, PIR schemes that are secure against  malicious severs
are very interesting.
In particular, outsoursing  PIR  to untrusted clouds in the modern age of cloud computing
is very interesting.
Both of our VC schemes  can provide easy solutions in outsourcing PIR.
 Using
$\Pi_{\rm pe}$, the client can outsource a degree $n$ polynomial $f(x)$  to the cloud,
where $f(i)=w_i$ for every $i\in[n]$.
To privately retrieve $w_i$, the client can execute
$\Pi_{\rm pe}$  with input $i$.
In this solution, the communication cost consists of
$O(\log n)$ group elements.
 Using $\Pi_{\rm mm}$, the client can represent the $w$ as
 a square matrix  $M=(M_{ij})$  of
order $\sqrt{n}$ and delegate $M$ to the cloud.
To privately retrieve a bit $M_{ij}$,
the client can execute $\Pi_{\rm mm}$ with input
$x\in \{0,1\}^{\sqrt{n}}$,
where $x_j=1$ and all the other bits are 0.
In this solution, the communication cost consists of
 $O(\sqrt{n})$ group elements.
 Note that in our outsourced  PIR schemes, the computation results
 always belong to $\{0,1\}\subseteq \mathbb{M}$. Therefore, the theoretical
 limitation we discussed in Section \ref{section:discussion} does not really affect
the application of  our VC schemes in outsourcing PIR.

\section{Conclusions}
\label{section:conclusions}

In this paper, we  constructed privacy preserving
 VC schemes for both univariate polynomial evaluation and matrix multiplication,
 which have useful applications in  outsourcing  PIR.
Our main tools  are the recently developed multilinear
maps. A theoretical limitation of our
constructions is that the results of the computations should belong to a
polynomial-size domain. Although this limitation does not really affect
their applications in outsourcing PIR,
it is still interesting to remove it in the future works.
We also note that our VC schemes are only privately verifiable.
It is also interesting to construct privacy preserving
VC schemes that are publicly verifiable.

\appendix

\section{Proof for Lemma \ref{lemma:sda}}
\label{appendix:sda}

{\bf Lemma 2.1}
{\em
If  ${\rm SDA}_i$ holds, then ${\rm SDA}_j$ holds for every $j=i+1,\ldots,k$.
}

\vspace{0.25cm}

Suppose  there is a PPT algorithm $\mathcal{A}$ such that
\begin{equation}
\label{equation:jth-subgroup-decision}
|\Pr[v\leftarrow G_j: \mathcal{A}(\Gamma, v)=1]-\Pr[v\leftarrow G_j:
 \mathcal{A}(\Gamma, v^q)=1]|\geq \epsilon
 \end{equation}
 for some integer $j\in\{i+1,\ldots, k\}$.
 We shall construct a PPT algorithm $\mathcal{B}$
 that breaks ${\rm SDA}_i$ with the same advantage,
 where the $\cal B$ is given a pair $(\Gamma, u)$
 and must decide whether $u\leftarrow G_i$
 or $u=\alpha^q$ for  $\alpha\leftarrow G_i$.
  On input
 $(\Gamma,  u)$, the algorithm  $\mathcal{B}$ computes $v=e(u,g_{j-i})$ and gives $(\Gamma, v)$ to
  $\mathcal{A}$.
Upon receiving the output bit $b$ of $\mathcal{A}$, the algorithm
 $\mathcal{B}$  will also output $b$.
Clearly,  $v$ will be uniformly distributed over $G_j$ if $u\leftarrow G_i$. On the other hand,
 $v=e(\alpha, g_{j-i})^q\in G_j$  will be  a random element of order $p$ if
 $u=\alpha^q$ for $\alpha\leftarrow G_i$.
 Due to (\ref{equation:jth-subgroup-decision}), the algorithm
$\mathcal{B}$ can distinguish the two cases for $v$ with advantage $\epsilon$
and therefore distinguish the two cases for $u$ with the same advantage $\epsilon$ (
i.e., break ${\rm SDA}_i$ with advantage $\epsilon$).
Since ${\rm SDA}_i$ holds, the $\epsilon$ must be negligible in the security
parameter $\lambda$. Therefore, ${\rm SDA}_j$ must  also hold for every $j=i+1,\ldots, k$.

\section{Proof for Lemma \ref{lemma:msdhs}}
\label{appendix:msdhs}

{\em {\bf Lemma 2.2}
If $(k,n)$-MSDH holds, then except for a negligible fraction of the
$k$-multilinear map instances $\Gamma_k\leftarrow {\cal G}(1^\lambda,k)$,
either $\Pr[\mathcal{A}(p,q, \Gamma_k, g_1, g_1^s,  \ldots,
g_1^{s^n})=g_k^{p/s}]
<{\sf neg}(\lambda)$  for any PPT algorithm $\cal A$ or $\Pr[\mathcal{A}(p,q, \Gamma_k,  g_1, g_1^s,  \ldots,
g_1^{s^n})= g_k^{q/s}]
<{\sf neg}(\lambda)$  for any PPT algorithm $\cal A$, where the probabilities are taken over
 $s\leftarrow \mathbb{Z}_N$
and ${\cal A}$'s random coins.
}

\vspace{0.25cm}
Firstly, under the $(k,n)$-MSDH assumption, for any PPT algorithm $\cal A$, we must have that
\begin{equation}
\label{equation:msdh-N}
\Pr[{\cal A}(p,q, \Gamma_k,g_1,g_1^s.\ldots, g_1^{s^n})=g_1^{1/s}]< {\sf neg}(\lambda),
\end{equation}
where the probability is taken over $\Gamma_k\leftarrow {\cal G}(1^\lambda,k)$,
 $s\leftarrow \mathbb{Z}_N$ and ${\cal A}$'s random coins. This is true because otherwise there would be a PPT algorithm
which can find $\alpha=0$, compute $g_1^{1/(s+\alpha)}$
with non-negligible probability and therefore break the $(k,n)$-MSDH assumption.

Suppose there are $\epsilon$  fraction ($\epsilon$ non-negligible) of the
 $k$-multilinear map instances $\Gamma_k\leftarrow {\cal G}(1^\lambda,k)$ for each of which  there are two PPT algorithms ${\cal A}_1$
 and ${\cal A}_2$ such that
\begin{equation}
\label{equation:msdh-pq}
\begin{split}
\Pr[\mathcal{A}_1(p,q, \Gamma_k, g_1, g_1^s,  \ldots,
g_1^{s^n})=g_k^{p/s}]{\rm ~and~}
\Pr[\mathcal{A}_2(p,q, \Gamma_k, g_1, g_1^s,  \ldots,
g_1^{s^n})=g_k^{q/s}]
\end{split}
\end{equation}
are both non-negligible,
where the probabilities are taken over
 $s\leftarrow \mathbb{Z}_N$ and the random coins of ${\cal A}_1$ and ${\cal A}_2$.
 We call such an instance $\Gamma_k$ {\em excellent}.
 Fix  an excellent instance $\Gamma_k$.
We say that $s\in \mathbb{Z}_N$ is {\em good} if
$\Pr[\mathcal{A}_1(p,q, \Gamma_k, g_1, g_1^s,  \ldots,
g_1^{s^n})=g_k^{p/s}]$ is non-negligible and
{\em nice} if $\Pr[\mathcal{A}_2(p,q, \Gamma_k, g_1, g_1^s,  \ldots,
g_1^{s^n})=g_k^{q/s}]$ is non-negligible, where the probabilities are taken over the
random coins of ${\cal A}_1$ and ${\cal A}_2$, respectively. For every $s\in\mathbb{Z}_N$, define $X_s$ and $Y_s$
to be two 0-1 random variables, where $X_s=1$ iff $s$ is good and
$Y_s=1$  iff $s$ is nice.  Due to (\ref{equation:msdh-pq}),
we must have that $\Pr[X_s=1]\triangleq \delta_1$ and $\Pr[Y_s=1]\triangleq \delta_2$ are both non-negligible, where the probabilities are
taken over $s\leftarrow \mathbb{Z}_N$. This so  because
otherwise both probabilities in (\ref{equation:msdh-pq}) would be negligible
as they must be taken over $s\leftarrow \mathbb{Z}_N$. Below we construct a simulator
$\cal S$ that simulates ${\cal A}_1$ and ${\cal A}_2$ and
contradicts to  (\ref{equation:msdh-N}).
The simulator $\cal S$ is given a tuple $(p,q,\Gamma_k,g_1, g_1^s, \ldots, g_1^{s^n})$
and required to compute $g_k^{1/s}$, where $s\leftarrow \mathbb{Z}_N$.
In order to do so the simulator $\cal S$ picks $s_1, s_2\leftarrow \mathbb{Z}_N$
independently. Let $\hat{s}_1=ss_1$ and $\hat{s}_2=ss_2$. It feeds ${\cal A}_1$ and ${\cal A}_2$ with
$(p,q,\Gamma_k, g_1, g_1^{\hat{s}_1}, \ldots, g_1^{\hat{s}_1^n})$
and $(p,q,\Gamma_k, g_1, g_1^{\hat{s}_2}, \ldots, g_1^{\hat{s}_2^n})$, respectively.
Note that the distributions of both $\hat{s}_1$ and $\hat{s}_2$ are
statistically close to the uniform distribution over $\mathbb{Z}_N$.
Therefore, with probability $\geq \delta_1-{\sf neg}(\lambda)$, $\hat{s}_1$ is good;
with probability $\geq \delta_2-{\sf neg}(\lambda)$,
$\hat{s}_2$ is nice.  As a consequence, with non-negligible  probabilities
${\cal A}_1$ and ${\cal A}_2$ will return $g_1^{p/\hat{s}_1}$ and $g_1^{q/\hat{s}_2}$,
respectively.  As a consequence, the simulator $\cal S$ can compute
$g_1^{p/s}$ and $g_1^{q/s}$ with non-negligible probability  since it knows
$s_1$ and $s_2$. Let $u,v$ be integers such that $up+vq=1$. Then the simulator
$\cal S$ can furthermore compute $g_1^{1/s}=(g_1^{p/s})^u(g_1^{q/s})^v$ with the same
(non-negligible) probability. In other words, we have that
$\Pr[{\cal S}(p,q, \Gamma_k,g_1,g_1^s.\ldots, g_1^{s^n})=g_1^{1/s}]$
is non-negligible for the fixed excellent instance $\Gamma_k$, where the probability is
taken over $s\leftarrow \mathbb{Z}_N$ and the random coins of $\cal S$.
Since the fraction $\epsilon$ of excellent  instances is also non-negligible,
we have that
 $\Pr[{\cal S}(p,q, \Gamma_k,g_1,g_1^s.\ldots, g_1^{s^n})=g_1^{1/s}]
 \geq \Pr[{\cal A}(p,q, \Gamma_k,g_1,g_1^s.\ldots, g_1^{s^n})=g_1^{1/s}
 |\Gamma_k{\rm~is~excellent}]\Pr[\Gamma_k{\rm~is~excellent}]$,
which is also non-negligible and contradicts to (\ref{equation:msdh-N}).
The lemma follows.

\section{Proof for Lemma \ref{lemma:tech}}
\label{appendix:tech}

{\bf Lemma 2.2}
{\em
Let $X$ and $Y$  be two  uniform random variables over  $\mathbb{Z}_N$.
 Then the random variable  $Z= pXY\bmod q$
is {\em statistically close} to the uniform   random variable
$U$ over $\mathbb{Z}_q$, i.e.,  we have that
$\sum_{\omega\in \mathbb{Z}_q}|\Pr[Z=\omega]-\Pr[U=\omega]|<{\sf neg}(\lambda).$
}

\vspace{0.25cm}

Let  $W=XY\bmod N$. Then $Z=pW\bmod q$.
We firstly  determine the distribution of $W$.
Since both $X$ and $Y$ are uniformly distributed over $\mathbb{Z}_N$,
it is easy to see that $$\Pr[W=w]=\frac{|\{(x,y)\in \mathbb{Z}_N^2: xy=w\bmod N\}|}{N^2}$$
for every $w\in \mathbb{Z}_N$. Therefore, it suffices to determine the number
${\cal N}_w$ of pairs
$(x,y)\in \mathbb{Z}_N^2$ such that $xy=w\bmod N$ for every $w\in \mathbb{Z}_N$.
Let $S_1=\{w: w =0\bmod p{\rm~and~}w=0\bmod q\}$,
$S_2=\{w: w \neq 0\bmod p{\rm~and~}w\neq 0\bmod q\}$,
$S_3=\{w: w =0\bmod p{\rm~and~}w\neq 0\bmod q\}$, and
$S_4=\{w: w\neq 0\bmod p{\rm~and~}w=0\bmod q\}$.
Clearly, we have that $|S_1|=1, |S_2|=N-p-q+1, |S_3|=q-1$ and $ |S_4|=p-1$.
Simple calculations show that
\begin{equation}
\label{equation:card}
{\cal N}_w=
\begin{cases}
4N-2p-2q+1, & {\rm if~} w\in S_1;\\
N-p-q+1,   & {\rm if~}   w\in S_2;\\
2N-2p-q+1, & {\rm if~} w\in S_3;\\
2N-p-2q+1, & {\rm if~} w\in S_4.
\end{cases}
\end{equation}
Next, we consider the map $\sigma: \mathbb{Z}_N\rightarrow \mathbb{Z}_q$ defined
by $\sigma(w)=pw\bmod q$ for every $w\in \mathbb{Z}_N$.
It is not hard to verify that
$\left. \sigma\right|_{S_1}$ and $\left. \sigma\right|_{S_4}$ are zero maps,
$\left. \sigma\right|_{S_2}$ is a $(p-1)$-to-1 map with range $\mathbb{Z}_q^*$, and
$\left. \sigma\right|_{S_3}$ is a 1-to-1 map with range $\mathbb{Z}_q^*$.  It follows that
$\sigma^{-1}(0)=S_1\cup S_4$, and
\begin{equation}
\label{equation:preimage}
\begin{split}
\sigma^{-1}(\omega)&\subseteq S_2\cup S_3,\\
|\sigma^{-1}(\omega)\cap S_2|&=p-1,\\|\sigma^{-1}(\omega)
\cap S_3|&=1
\end{split}
\end{equation}
for every $\omega\in \mathbb{Z}_q^*$. Due to (\ref{equation:card}) and (\ref{equation:preimage}),
the distribution of  $Z$ can be described by
\begin{equation*}
\begin{split}
\Pr[Z=0]&=\frac{1}{N^2}((4N-2p-2q+1)|S_1|+(2N-p-2q+1)|S_4|)=\frac{p(2N-p)}{N^2} {\rm~and~}\\
\Pr[Z=\omega]&=\frac{1}{N^2}((N-p-q+1)(p-1)+(2N-2p-q+1))=\frac{p^2(q-1)}{N^2},
\end{split}
\end{equation*}
where $\omega\in \mathbb{Z}_q^*$ is arbitrary. It follows that
$$\sum_{\omega\in \mathbb{Z}_q}\big|\Pr[Z=\omega]-\Pr[U=\omega]\big|=\left|\frac{p(2N-p)}{N^2}-
\frac{1}{q}\right|+
\left|\frac{p^2(q-1)}{N^2}-
\frac{1}{q}\right|(q-1)=\frac{2(q-1)}{q^2},$$
which is negligible in the security parameter
$\lambda$ as $q$ is a $\lambda$-bit prime.

\section{Proof for Lemma \ref{lemma:3-co-cdhs}}
\label{appendix:3-co-cdhs}

{\bf Lemma 2.3}
{\em
If {3-MDDH} holds, then  {3-co-CDHS} also holds.
}

\vspace{0.25cm}

 Suppose that the 3-co-CDHS does not hold. Then
 there is a PPT adversary
${\cal A}$ such that
\begin{equation}
\label{equation:adv}
\Pr[{\cal A}(p,q,\Gamma, h_1^a, h_2^b)=h_2^{ab}]
\geq \epsilon,
\end{equation}
where $\epsilon$ is non-negligible,  and
the probability is taken over $a,b\leftarrow \mathbb{Z}_N$ and
the random coins of ${\cal A}$. We shall construct a  PPT adversary ${\cal B}$ that breaks the 3-MDDH. The adversary
${\cal B}$ is given $(p,q,\Gamma)$ along with five
group elements $t=g_1^s, \alpha=g_1^a, \beta=g_1^b, \gamma=g_1^c$ and $h\leftarrow G_3$,
where $s,a,b,c\leftarrow  \mathbb{Z}_N$.
It must decide whether  $h=g_3^{sabc}$ or not.
In order to do so,  ${\cal B}$   computes $u=\alpha^p,
 v=e(\beta,\gamma)^p=g_2^{pbc}$,  gives
$(p,q,\Gamma,u,v)$ to ${\cal A}$ and receives $w$ from
${\cal A}$.
 Clearly, $u\in G_1$ is   is a random element of order  $q$.
 On the other hand,  due to Lemma \ref{lemma:tech},
the distribution of $v=g_2^{pbc}$ is statistically close to the uniform distribution
over the order $q$ subgroup of $G_2$.
Therefore, equation (\ref{equation:adv})  implies that
$$\Pr[w=g_2^{pabc}]\geq \epsilon-{\sf neg}(\lambda),$$ where the probability is taken over
 $a,b,c\leftarrow \mathbb{Z}_N$ and
the random coins of ${\cal A}$. Given $t\in G_1$ and $w\in G_2$, ${\cal B}$
will compute $\sigma=e(t,w)$ and compare with $\tau=h^p$. At last, ${\cal B}$
will output 1 if $\sigma=\tau$ and a random bit $\psi\in \{0,1\}$ otherwise.
Let $\mathbb{E}_1$ be the event that ${\cal B}(p,q,\Gamma, t,\alpha,\beta,\gamma,h)=1$
when $h=g_3^{sabc}$. Then
 \begin{equation}
 \label{equation:h-not-random}
 \begin{split}
 \Pr[\mathbb{E}_1]&=
 \Pr[\mathbb{E}_1| w=g_2^{pabc}]
 \Pr[w=g_2^{pabc}]+\Pr[\mathbb{E}_1|
 w\neq g_2^{pabc}]
 \Pr[w\neq g_2^{pabc}]\\
 &\geq \Pr[w=g_2^{pabc}]+\frac{1}{2}(1- \Pr[w=g_2^{pabc}])\geq \frac{1}{2}+\frac{1}{2}\epsilon-{\sf neg}(\lambda).
 \end{split}
 \end{equation}
Let $\mathbb{E}_2$ be the event that ${\cal B}(p,q,\Gamma,
   t,\alpha,\beta,\gamma,h)=1$
when $h\leftarrow G_3$ is a random group element. Then
  \begin{equation}
  \label{equation:h-random}
 \begin{split}
 \Pr[\mathbb{E}_2]&=
 \Pr[\mathbb{E}_2| w=g_2^{pabc}]
 \Pr[w=g_2^{pabc}]+\Pr[\mathbb{E}_2|w\neq g_2^{pabc}]
 \Pr[w\neq g_2^{pabc}].
 \end{split}
 \end{equation}
Let $h=g_3^{\delta}$ for  $\delta\leftarrow \mathbb{Z}_N$.
Let $\eta\in \mathbb{Z}_N$ be fixed.
It is not hard to see that  $\Pr[h^p=g_3^{p\eta}]=
\Pr[\delta\equiv \eta\bmod q]=1/q$ and $\Pr[h^p=g_3^{\eta}]=
\Pr[p\delta\equiv \eta\bmod N]\leq 1/q$ where the probability is taken over
$\delta\in \mathbb{Z}_N$. It follows that
\begin{equation}
\label{equation:h-random-1}
\begin{split}
\Pr[\mathbb{E}_2| w=g_2^{pabc}]
&=
\Pr[h^p=g_3^{psabc}]\leq \frac{1}{q};\\
\Pr[\mathbb{E}_2|w\neq g_2^{pabc}]
 &=\Pr[\mathbb{E}_2, h^p=e(t,w)|w\neq g_2^{pabc}]+\Pr[\mathbb{E}_2,
 h^p\neq e(t,w)|w\neq g_2^{pabc}]\\
 &=\Pr[h^p=e(t,w)|w\neq g_2^{pabc}]+\frac{1}{2}\Pr[h^p\neq e(t,w)|w\neq g_2^{pabc}]\\
 &=\frac{1}{2}(1+\Pr[h^p=e(t,w)|w\neq g_2^{pabc}])
 \leq \frac{1}{2}+\frac{1}{2q}.
\end{split}
\end{equation}
Due to  (\ref{equation:h-random}) and (\ref{equation:h-random-1}), we have that
\begin{equation}
\label{equation:h-random-2}
\begin{split}
\Pr[\mathbb{E}_2]&\leq \frac{1}{q} \Pr[ w=g_2^{pabc}]
+\left(\frac{1}{2}+\frac{1}{2q}\right)(1- \Pr[w=g_2^{pabc}])\\
&\leq \frac{1}{q}+\left(\frac{1}{2}+\frac{1}{2q}\right)(1-\epsilon+{\sf neg}(\lambda))
\leq \frac{1}{2}-\frac{1}{2}\epsilon+{\sf neg}(\lambda),
\end{split}
\end{equation}
where the second inequality follows from $\Pr[w=g_2^{pabc}]\geq \epsilon-{\sf neg}(\lambda)$
and the third inequality follows from the fact that $1/q$ is negligible in $\lambda $.
Due to (\ref{equation:h-not-random}) and
(\ref{equation:h-random-2}), we have that
$$|\Pr[\mathbb{E}_1]-\Pr[\mathbb{E}_2]|\geq \epsilon-{\sf neg}(\lambda),$$
which says that 3-MDDH assumption can be broken  with advantage
at least $\epsilon-{\sf neg}(\lambda)$, which is non-negligible
as long as $\epsilon$ is non-negligible.
Hence, the 3-co-CDHS assumption must hold as long as the  3-MDDH
assumption holds.

\section{Proof for Lemma \ref{theorem:privacy-PE}}
\label{appendix:privacy-pe}

{\bf Lemma 3.3}
{\em
If {SDA} holds for $\Gamma$, then the scheme
$\Pi_{\rm pe}$ achieves the input privacy.
}

\vspace{0.25cm}

Let $\Gamma$ be the multilinear map instance
in $\Pi_{\rm pe}$. Given any input $\alpha$,
the only message  that contains information about $\alpha$ is $\sigma=(\sigma_1,\ldots, \sigma_k)$,
which is sent to the server by the client.
 Suppose that there is an adversary
${\cal A}$ that  breaks the input privacy of our scheme. Then
${\cal A}$ must succeed  with a non-negligible
advantage $\epsilon$   in
${\sf Exp}_{\cal A}^{\sf Pri}(\Pi,f,\lambda)$, i.e.,
$\Pr[{\sf Exp}_{\cal A}^{\sf Pri}(\Pi,f,\lambda)=1]\geq \frac{1}{2}+\epsilon$.
Below we construct a simulator $\cal S$ that breaks the
 semantic security of ${\rm BGN}_{2k+1}$  with  non-negligible advantage
 $\geq \epsilon/k$. In the semantic security game of ${\rm BGN}_{2k+1}$, the
 challenger stores the secret key $p$ locally and gives
 a public key $(\Gamma, g_1, h)$ to $\cal S$.
To break ${\rm BGN}_{2k+1}$, the simulator $\cal S$ simulates $\cal A$ as below:
\begin{itemize}
\item[(A)] Mimic ${\sf ProbGen}$: Picks a polynomial $f(x)=f_0+f_1x+\cdots+f_nx^n$. Picks $s\leftarrow
 \mathbb{Z}_N$ and computes $t=g_1^{f(s)}$.
Stores  $sk=(s,t)$ and gives
 $pk=(\Gamma, g_1,h, g_1^s, g_1^{s^2}, \ldots, g_1^{s^{2^{k-1}}},f)$
 to $\cal A$. Note that $sk$ does not include $p,q$ as components of $sk$.
 This is because $p$ and $q$ are neither known to the simulator
 $\cal S$ nor used by the simulator $\cal S$.

\item[(B)] Mimic ${\sf PubProbGen}(sk,\cdot)$: Upon receiving a query $\alpha\in \mathbb{D}$,
 picks $r_\ell\leftarrow \mathbb{Z}_N$
and computes $\sigma_\ell=g_1^{\alpha^{2^{\ell-1}}}h^{r_\ell}$ for every $\ell\in [k]$.
Then gives $\sigma=(\sigma_1,\ldots, \sigma_k)$ to $\cal A$.
\item[(C)] Upon receiving $\alpha_0, \alpha_1 \in \mathbb{D}$  from  $\cal A$,
the simulator $\cal S$  picks $i\leftarrow \{0,1,\ldots,k-1\}$ and sends
$\beta_0=\alpha_0^{2^i}$ and $\beta_1=\alpha_1^{2^i}$ to the  challenger.
 The challenger will pick $b\leftarrow \{0,1\}$ and
sends ${\sf Enc}(\beta_b)$ to $\cal S$.
Upon receiving ${\sf Enc}(\beta_b)$ from the challenger.
The simulator $\cal S$ gives
$$Z=({\sf Enc}(\alpha_1), \ldots, {\sf Enc}(\alpha_1^{2^{i-1}}), {\sf Enc}(\beta_b),
{\sf Enc}(\alpha_0^{2^{i+1}}), \ldots, {\sf Enc}(\alpha_0^{2^{k-1}}))$$
to $\cal A$ and learns a bit $b^\prime$ in return.
\item[(D)] The simulator $\cal S$  outputs $\hat{b}=1$ if $b^\prime =1$ and
$\hat{b}=0$ otherwise. It  wins if $\hat{b}=b$.
\end{itemize}
For every $j\in \{0,1,\ldots, k\}$,  we define the following probability ensemble
$$Z_j=({\sf Enc}( \alpha_1),\ldots, {\sf Enc}( \alpha_1^{2^{j-1}}), {\sf Enc}( \alpha_0^{2^j}),
\ldots, {\sf Enc}(\alpha_0^{2^{k-1}})).$$
 Then
$
Z_0=({\sf Enc}( \alpha_0),{\sf Enc}( \alpha_0^{2}),
\ldots, {\sf Enc}(\alpha_0^{2^{k-1}}))$ and $
Z_k=({\sf Enc}( \alpha_1), {\sf Enc}( \alpha_1^{2}),
\ldots, {\sf Enc}(\alpha_1^{2^{k-1}}))
$.
The  inequality $\Pr[{\sf Exp}_{\cal A}^{\sf Pri}(\Pi,f,\lambda)=1]\geq
\frac{1}{2}+\epsilon$ implies that
  $\Pr[{\cal A}(pk,Z_k)=1]-
\Pr[{\cal A}(pk, Z_0)=1]\geq 2\epsilon$.
Note that $Z=Z_i$ when $b=0$ and
 $Z=Z_{i+1}$ when $b=1$.
%
Let $\mathbb{E}$
be the event that $\hat{b}=b$. Then
\begin{equation*}
\begin{split}
\Pr[\mathbb{E}]&=
\sum_{l=0}^{k-1}\Pr[\mathbb{E}|i=l]\Pr[i=l]=
\frac{1}{k}\sum_{l=0}^{k-1} (\Pr[\mathbb{E}|i=l,b=0]\Pr[b=0]+\Pr[\mathbb{E}|i=l,b=1]\Pr[b=1])\\
&=
\frac{1}{2k}\sum_{l=0}^{k-1} (\Pr[b^\prime=0|i=l,b=0]+\Pr[b^\prime=1|i=l,b=1])\\
&=
\frac{1}{2k}\sum_{l=0}^{k-1} (1-\Pr[{\cal A}(pk,Z_l)=1]+\Pr[{\cal A}(pk,Z_{l+1})=1])\\
&=\frac{1}{2}+\frac{1}{2k}(\Pr[{\cal A}(pk,Z_k)=1]-\Pr[{\cal A}(pk,Z_0)=1])
\geq \frac{1}{2}+\frac{\epsilon}{k},
\end{split}
\end{equation*}
i.e.,  ${\cal B}$
breaks the semantic security of ${\rm BGN}_{2k+1}$ with a non-negligible advantage $\epsilon/k$,
 which is not true when SDA holds. Hence, $\Pi_{\rm pe}$ achieves input privacy.

\section{Proof for Lemma \ref{theorem:privacy-MM}}

{\bf Lemma 3.6}
{\em
If  the {SDA}  for $\Gamma$ holds,
then  $\Pi_{\rm mm}$
achieves the input privacy.
}

\vspace{0.25cm}
We define  three games ${\sf G}_0,{\sf G}_1$ and ${\sf G}_2$ as below:
\begin{itemize}
\item[${\sf G}_0:$] this  is the standard security game
              ${\sf Exp}_{\cal A}^{\sf Pri}(\Pi,M,\lambda)$ defined in
              {Fig.} \ref{figure:experiment}.
\item[${\sf G}_1:$] the only difference between this game and ${\sf G}_0$
is a change to  {\sf ProbGen}. For any $(x_1,\ldots,x_n)$ queried  by the adversary, instead of computing
$\tau$ using the efficient {\sf CFE} algorithm, the inefficient evaluation of $\tau_i$ is used,
i.e., $\tau_i=\prod_{j=1}^n e({\sf F}_K(i,j)^{x_j}, g_2^{p})$ for every $i\in[n]$.
\item[${\sf G}_2:$] the only difference between this game and ${\sf G}_1$ is that
the matrix $T$ is computed as $T_{ij}=g_1^{p^2aM_{ij}}\cdot R_{ij}$,
where $R_{ij}\leftarrow G_1$  for every $i,j\in[n]$.
\end{itemize}
For every $i\in \{0,1,2\}$, we denote by
${\sf G}_i({\cal A})$ the output of game $i$ when it is run with an adversary ${\cal A}$.
The proof of the theorem proceeds by a standard hybrid argument, and is obtained by combining
the proofs of the following three claims.

\vspace{1mm}
\noindent\underline{\rm Claim 1}. We have that $\Pr[{\sf G}_0({\cal A})=1]=
\Pr[{\sf G}_1({\cal A})=1]$.

The only difference between ${\sf G}_1$ and  ${\sf G_0}$
is in the
computation of $\tau$.   Due to the correctness of the ${\sf CFE}$
algorithm,  such difference does not change the distribution of the
values  $\tau$ returned to the adversary. Therefore,
 the probabilities that ${\cal A}$
wins in both games are identical.

\vspace{0.1cm}
\noindent\underline{\rm Claim 2}. We have that $|\Pr[{\sf G}_1({\cal A})=1]-
\Pr[{\sf G}_2({\cal A})=1]|<{\sf neg}(\lambda)$.

The only difference between  ${\sf G}_2$ and ${\sf G_1}$ is
 that we replace the pseudorandom group elements  ${\sf F}_K(i,j)$ with truly random
 group elements $R_{ij}\leftarrow G_1$ for every $i,j\in [n]$.
Clearly, if $|\Pr[{\sf G}_1({\cal A})=1]-
\Pr[{\sf G}_2({\cal A})=1]|$ is non-negligible, we can construct an simulator ${\cal B}$ that
simulates  ${\cal A}$ and breaks the pseudorandom property of  $\sf PRF$
with a non-negligible advantage.

\vspace{0.1cm}
\noindent\underline{\rm Claim 3}. We have that $\Pr[{\sf G}_2({\cal A})=1]<
{\sf neg}(\lambda)$.

Let $\Gamma$ be the multilinear map instance
in $\Pi_{\rm mm}$. Given any input $x$,
the only message  that contains information about $x$ is $\sigma=(\sigma_1,\ldots, \sigma_n)$,
which is sent to the server by the client.
 Suppose that there is an adversary
${\cal A}$ that  breaks the input privacy of our scheme. Then
${\cal A}$ must succeed  with a non-negligible
advantage $\epsilon$   in
${\sf Exp}_{\cal A}^{\sf Pri}(\Pi,M,\lambda)$, i.e.,
$\Pr[{\sf Exp}_{\cal A}^{\sf Pri}(\Pi,M,\lambda)=1]\geq \frac{1}{2}+\epsilon$.
Below we construct a simulator $\cal S$ that breaks the
 semantic security of ${\rm BGN}_{3}$  with  non-negligible advantage
 $\geq \epsilon/n$. In the semantic security game of ${\rm BGN}_{3}$, the
 challenger stores the secret key $p$ locally and gives
 a public key $(\Gamma, g_1, h)$ to $\cal S$.
To break ${\rm BGN}_{3}$, the simulator $\cal S$ simulates $\cal A$ as below:
\begin{itemize}
\item[(A)] Mimic ${\sf ProbGen}$: Picks a
square matrix $M=(M_{ij})$ of order $n$.
 Picks $a\leftarrow
 \mathbb{Z}_N$. Computes $T_{ij}=g_1^{p^2aM_{ij}}R_{ij}$ for every
 $(i,j)\in [n]^2$, where $R_{ij}\leftarrow G_1$.
Stores  $sk=(a,\eta)$ and gives
 $pk=(\Gamma, g_1,h, M,T)$
 to $\cal A$, where $\eta=g_3^{p^2a}$.  Note that $sk$ does not include $p,q$ as components of $sk$.
 This is because $p$ and $q$ are neither known to the simulator
 $\cal S$ nor used by the simulator $\cal S$.

\item[(B)] Mimic ${\sf PubProbGen}(sk,\cdot)$: Upon receiving a query $x=(x_1,\ldots,x_n)\in \mathbb{D}$,
 picks $r_j\leftarrow \mathbb{Z}_N$
and computes $\sigma_j=g_1^{x_j}h^{r_j}$ for every $j\in [n]$.
Then gives $\sigma=(\sigma_1,\ldots, \sigma_n)$ to $\cal A$.
\item[(C)] Upon receiving $x^0, x^1 \in \mathbb{D}$  from  $\cal A$,
the simulator $\cal S$ picks $i\leftarrow [n]$ and sends
$\beta_0=x^0_i$ and $\beta_1=x^1_i$ to the  challenger (Here without loss of generality,
we can suppose that the Hamming distance between $x^0$ and $x^1$ is $n$. Otherwise, we can
ignore  their equal components and  only consider
the different components.).
 The challenger will pick $b\leftarrow \{0,1\}$ and
sends ${\sf Enc}(\beta_b)$ to $\cal S$.
Upon receiving ${\sf Enc}(\beta_b)$ from the challenger.
The simulator $\cal S$ gives
$$Z=({\sf Enc}(x^1_1), \ldots, {\sf Enc}(x^1_{i-1}), {\sf Enc}(\beta_b),
{\sf Enc}(x^0_{i+1}), \ldots, {\sf Enc}(x^0_n))$$
to $\cal A$ and learns a bit $b^\prime$ in return.
\item[(D)] The simulator $\cal S$  outputs $\hat{b}=1$ if $b^\prime =1$ and
$\hat{b}=0$ otherwise. It  breaks ${\rm BGN}_3$ if $\hat{b}=b$.
\end{itemize}
For every $j\in \{0,1,\ldots,n\}$,  we define the following probability ensemble
$$Z_j=({\sf Enc}(x^1_1), \ldots, {\sf Enc}(x^1_{j}), {\sf Enc}(x^0_{j+1}), \ldots, {\sf Enc}(x^0_n)).$$
 Then
$
Z_0=({\sf Enc}(x^0_1),{\sf Enc}(x^0_2),
\ldots, {\sf Enc}(x^0_n))$ and $
Z_n=({\sf Enc}(x^1_1), {\sf Enc}(x^1_2),
\ldots, {\sf Enc}(x^1_n))
$.
The  inequality $\Pr[{\sf Exp}_{\cal A}^{\sf Pri}(\Pi,M,\lambda)=1]\geq
\frac{1}{2}+\epsilon$ implies that
  $\Pr[{\cal A}(pk,Z_n)=1]-
\Pr[{\cal A}(pk, Z_0)=1]\geq 2\epsilon$.
Note that $Z=Z_{i-1}$ when $b=0$ and
 $Z=Z_{i}$ when $b=1$.
%
Let $\mathbb{E}$
be the event that $\hat{b}=b$. Then
\begin{equation*}
\begin{split}
\Pr[\mathbb{E}]&=
\sum_{l=1}^{n}\Pr[\mathbb{E}|i=l]\Pr[i=l]=
\frac{1}{n}\sum_{l=1}^{n} (\Pr[\mathbb{E}|i=l,b=0]\Pr[b=0]+\Pr[\mathbb{E}|i=l,b=1]\Pr[b=1])\\
&=
\frac{1}{2n}\sum_{l=1}^{n} (\Pr[b^\prime=0|i=l,b=0]+\Pr[b^\prime=1|i=l,b=1])\\
&=
\frac{1}{2n}\sum_{l=1}^{n} (1-\Pr[{\cal A}(pk,Z_{l-1})=1]+\Pr[{\cal A}(pk,Z_{l})=1])\\
&=\frac{1}{2}+\frac{1}{2n}(\Pr[{\cal A}(pk,Z_n)=1]-\Pr[{\cal A}(pk,Z_0)=1])
\geq \frac{1}{2}+\frac{\epsilon}{n},
\end{split}
\end{equation*}
i.e.,  ${\cal B}$
breaks the semantic security of ${\rm BGN}_{3}$ with a non-negligible advantage $\epsilon/n$,
 which is not true when SDA holds. Hence, $\Pi_{\rm mm}$ achieves input privacy.

\section{Privacy Preserving VC Schemes}
\label{appendix:modified-vc}

This section shows the  modifications of our VC schemes
$\Pi_{\rm pe}$ and $\Pi_{\rm mm}$  that achieve both input and function privacy.

\begin{figure}[H]
\begin{center}
\begin{boxedminipage}{17cm}
\begin{itemize}
\item  $\mathsf{KeyGen}(1^\lambda, f(x))$:
Pick
 $\Gamma=(N, G_1,\ldots,   G_{2k+2},   e,  g_1, \ldots,  g_{2k+2})\leftarrow
 {\cal G}(1^\lambda, 2k+2)$.
Pick  $ s\leftarrow \mathbb{Z}_N$ and compute
 $t=g_1^{f(s)}$.
 Pick
$u\leftarrow G_1$ and compute $h=u^q$, where $u=g_1^{\delta}$ for  an integer
$\delta\in \mathbb{Z}_N$. Set up ${\rm BGN}_{2k+2}$  with public key
$(\Gamma,g_1, h)$ and secret key $p$.
For every $i\in\{0,1,\ldots,n\}$, pick $v_i\leftarrow \mathbb{Z}_N$ and
compute $\gamma_i=g_1^{f_i}h^{v_i}$.
Output  $sk=(p,q, s, t)$ and
$pk=(\Gamma, g_1,h,  g_1^s, g_1^{s^2},  \ldots, g_1^{s^{2^{k-1}}}, \gamma)$, where $\gamma=(\gamma_0,\ldots, \gamma_n)$.
\item ${\sf ProbGen}(sk,\alpha)$:   For every $\ell\in [k]$, pick
 $r_\ell \leftarrow \mathbb{Z}_N$ and compute
 $\sigma_\ell=g_1^{\alpha^{2^{\ell-1}}}h^{r_\ell}$. Output
 $\sigma=(\sigma_1,\ldots, \sigma_k)$ and  $\tau=\perp$
  ($\tau$ is  not used).
\item ${\sf Compute}(pk,\sigma)$: Compute
$\rho_i=g_{k}^{\mu_i}$ for  every   $i\in \{0,1,\ldots,n\}$ using the technique in Section
\ref{section:BGN}.
Compute $\rho^\prime_i=e(\gamma_i, \rho_i)=g_{k+1}^{\mu^\prime_i}$, where
$\mu^\prime_i=(f_i+q\delta v_i)\mu_i$.
Compute  $\rho=\prod_{i=0}^n \rho^\prime_i$.
Compute $\pi_{ij}=g_{2k}^{\nu_{ij}}$
 using the technique in
Section \ref{section:BGN} for every $i\in \{0,1,\ldots,n-1\}$ and
$j\in\{0,1,\ldots,i\}$.
Compute $\pi^\prime_{ij}=e(\gamma_{i+1}, \pi_{ij})=g_{2k+1}^{\nu^\prime_{ij}}$, where
$\nu^\prime_{ij}=(f_{i+1}+q\delta v_{i+1})\nu_{ij}$.
Set $\pi=\prod_{i=0}^{n-1}\prod_{j=0}^i \pi^\prime_{ij}$.Output  $\rho$ and  $\pi$.
\item ${\sf Verify}(sk, \tau, \rho,\pi)$: Compute  the $y\in \mathbb{Z}_q$ such that
$\rho^p=(g_{k+1}^p)^y$.
If the equality
$
e\big(t/g_1^{y}  , g_{2k+1}^{p}\big)=
e\big(g_1^s/g_1^\alpha, \pi^{p}\big)
$
 holds,
output $y$; otherwise,
output $\perp$.

\end{itemize}
\end{boxedminipage}
\vspace{-2mm}
\caption{Univariate polynomial evaluation ($\Pi_{\rm pe}^\prime$)}
\label{figure:pp-vc-pe}
\end{center}
\end{figure}

\begin{figure}[H]
\begin{center}
\vspace{0cm}
\begin{boxedminipage}{17cm}
\begin{itemize}
\item ${\sf KeyGen}(1^\lambda,M)$:
Pick $\Gamma=(N,G_1,G_2, G_3,e,g_1,g_2,g_3)\leftarrow {\cal G}(1^\lambda,3)$.
Consider the ${\sf PRF}_{\rm dlin}$  in Section \ref{section:aprf}.
Run  ${\sf KG}(1^\lambda,n)$
and pick a  secret key $K$.
Pick  $a\leftarrow \mathbb{Z}_N$ and compute
 $T_{ij}=g_1^{p^2aM_{ij}}\cdot {\sf   F}_K(i,j)$ for every  $(i,j)\in[n]^2$.
 Pick  $u\leftarrow G_1$ and compute
$ h=u^q$,  where $u=g_1^{\delta}$ for  an integer
$\delta\in \mathbb{Z}_N$.
Set up ${\rm BGN}_{3}$ with public key
$(\Gamma,g_1,h)$ and secret key $p$.
For every $(i,j)\in [n]^2$, pick $v_{ij}\leftarrow \mathbb{Z}_N$
and compute $\gamma_{ij}=g_1^{M_{ij}}h^{v_{ij}}$.
Output  $sk=(p,q,K, a, \eta)$ and   $pk=(\Gamma, g_1,h, \gamma,T)$, where $\eta=g_3^{p^2 a}$ and $\gamma=(\gamma_{ij})$.

\item ${\sf ProbGen}(sk,x)$:
For every $j\in[n]$, pick  $r_j\leftarrow \mathbb{Z}_N$
and compute  $\sigma_j=g_1^{x_j}h^{r_j}$.
For every $i\in[n]$, compute
 $\tau_i= e(\prod_{j=1}^n{\sf F}_K(i,j)^{x_j}, g_2^{p})$ using the efficient ${\sf CFE}$
 algorithm in Section \ref{section:aprf}.
Output  $\sigma=(\sigma_1,\ldots,\sigma_n)$ and
 $\tau=(\tau_1,\ldots, \tau_n)$.
\item ${\sf Compute}(pk,\sigma)$:
Compute
 $\rho_i=\prod_{j=1}^n e(\gamma_{ij}, \sigma_j)$ and
 $\pi_i=\prod_{j=1}^n e(T_{ij},\sigma_j)$ for every $i\in[n]$. Output
  $\rho=(\rho_1,\ldots,\rho_n)$ and
 $\pi=(\pi_1,\ldots, \pi_n)$.
\item ${\sf Verify}(sk,\tau, \rho,\pi)$:
For every $i\in[n]$, compute  $y_i$ such that $\rho_i^p=(g_2^p)^{y_i}$.
If
$
e(\pi_i, g_1^p)=\eta^{py_i}\cdot \tau_i
$
 for every $i\in[n]$, then output $y=(y_1,\ldots,y_n)$;
otherwise, output $\perp$.
\end{itemize}
\end{boxedminipage}
\vspace{-2mm}
\caption{Matrix multiplication ($\Pi_{\rm mm}^\prime$)}
\label{figure:pp-vc-mm}
\end{center}
\end{figure}

\end{document}